\documentclass[copyright,creativecommons]{eptcs}

\usepackage{amsmath,amssymb,amsthm}
\usepackage{bm}
\usepackage[capitalize]{cleveref}
\usepackage{mathtools}
\usepackage{tikz,tikz-cd}
\usepackage{underscore}
\usepackage{quiver}

\theoremstyle{plain}
\newtheorem{theorem}{Theorem}[section]
\newtheorem{proposition}[theorem]{Proposition}
\newtheorem{lemma}[theorem]{Lemma}

\theoremstyle{definition}
\newtheorem{definition}[theorem]{Definition}

\newtheorem{example}[theorem]{Example}

\newtheorem{remark}[theorem]{Remark}

\newcommand{\A}{\mathcal{A}}
\newcommand{\C}{\mathcal{C}}
\newcommand{\D}{\mathcal{D}}
\newcommand{\G}{\mathcal{G}}
\newcommand{\M}{\mathcal{M}}
\newcommand{\E}{\mathcal{E}}
\newcommand{\TC}{\mathbf{1}}

\newcommand{\Set}{\mathbf{Set}}
\newcommand{\Poset}{\mathbf{Poset}}
\newcommand{\Grade}[2]{\mathbf{Grade}_{#1}#2}

\DeclareMathOperator{\colim}{colim}
\newcommand{\inj}[1]{\mathrm{in}_{#1}}

\newcommand{\Mor}{\mathbf{Mor}}
\newcommand{\Iso}{\mathbf{Iso}}
\newcommand{\Inj}{\mathbf{Inj}}
\newcommand{\Surj}{\mathbf{Surj}}
\newcommand{\Full}{\mathbf{Full}}

\newcommand{\tensor}{\otimes}
\newcommand{\tensorI}{\mathsf{I}}

\newcommand{\gtensor}{\odot}
\newcommand{\gI}{I}
\newcommand{\ctensor}{\boxdot}
\newcommand{\cI}{\mathsf{J}}

\newcommand{\ot}{\tensor}
\newcommand{\I}{{\tensorI}}
\newcommand{\lam}{\lambda}
\newcommand{\al}{\alpha}
\newcommand{\bt}{\ctensor}
\newcommand{\J}{\cI}

\newcommand{\T}{\mathsf{T}}
\newcommand{\monoidal}[1]{\mathbb{#1}}
\newcommand{\lax}[1]{\mathsf{#1}}

\newcommand{\Equiv}[1]{\mathrm{Equiv}_#1}

\newcommand{\mkS}[1]{\mathbf{S}[#1]}
\newcommand{\mkSigma}[1]{\boldsymbol{\Sigma}[#1]}

\newcommand{\oper}{\phi}
\newcommand{\goper}{\psi}

\newcommand{\slice}[2]{#1/#2}
\newcommand{\slicedom}[2]{#2_{#1}}
\newcommand{\slicemor}[2]{\mathrm{snd}}

\newcommand{\Cart}[2]{[#1, #2]_{\mathrm{cart}}}
\newcommand{\CartNt}[2]{[#1, #2]_{\mathrm{cartnt}}}

\renewcommand{\exp}[2]{#1 \Rightarrow #2}
\newcommand{\id}{\mathrm{id}}
\newcommand{\compose}{\circ}
\newcommand{\Id}{\mathrm{Id}}
\newcommand{\fcomp}{\cdot}
\newcommand{\natto}{\Rightarrow}
\newcommand{\tto}{\Rightarrow}
\newcommand{\xto}[1]{\xrightarrow{#1}}
\newcommand{\iso}{\cong}
\renewcommand{\equiv}{\simeq}
\newcommand{\epito}{\twoheadrightarrow}
\newcommand{\monoto}{\rightarrowtail}
\newcommand{\str}{\mathrm{str}}

\newcommand{\append}{\mathbin{+\mkern-5mu+}}
\newcommand{\xs}{\mathit{xs}}

\newcommand{\Nat}{\mathbb{N}}

\newcommand{\eps}{\varepsilon}

\title{Canonical Gradings of Monads}
\author{%
  Flavien Breuvart
  \institute{LIPN, Universit\'e Sorbonne Paris Nord, France}
  \email{flavien.breuvart@lipn.univ-paris13.fr}
  \and
  Dylan McDermott
  \institute{Reykjavik University, Iceland}
  \email{dylanm@ru.is}
  \and
  Tarmo Uustalu
  \institute{Reykjavik University, Iceland}
  \institute{Tallinn University of Technology, Estonia}
  \email{tarmo@ru.is}}
\def\titlerunning{Canonical Gradings of Monads}
\def\authorrunning{Flavien Breuvart, Dylan McDermott \& Tarmo Uustalu}

\providecommand{\event}{ACT 2022}

\begin{document}
\maketitle

\vspace*{-5mm}

\begin{abstract}
  We define a notion of grading of a monoid $\T$ in a monoidal category
  $\C$, relative to a class of morphisms $\M$ (which provide a notion
  of $\M$-subobject). We show that, under reasonable conditions
  (including that $\M$ forms a factorization system), there is a
  canonical grading of $\T$. Our application is to graded monads and
  models of computational effects. We demonstrate our results by
  characterizing the canonical gradings of a number of monads, for
  which $\C$ is endofunctors with composition. We also show that we
  can obtain canonical grades for algebraic operations.
\end{abstract}

\vspace*{-5mm}

\section{Introduction}

This paper is motivated by quantitative modelling of computational
effects from mathematical programming semantics. It is standard in
this domain to model notions of computational effect, such as
nondeterminism or manipulation of external state, by (strong) monads
\cite{moggi1989computational}. In many applications, however, it is useful to be able to
work with quantified effects, e.g., how many outcomes a computation
may have, or to what degree it may read or overwrite the state. This
is relevant, for example, for program optimizations or analyses to
assure that a program can run within allocated
resources. Quantification of effectfulness is an old idea and goes
back to type-and-effect systems \cite{lucassen1988polymorphic}. Mathematically, notions of
quantified effect can be modelled by graded (strong) monads \cite{Smirnov2008,mellies2012parametric,katsumata2014}.

It is natural to ask if there are systematic ways for refining a
non-quantitative model of some effect into a quantitative version,
i.e., for producing a graded monad from a monad. In this paper, we
answer this question in the affirmative.  We show how a monad on a
category can be graded with any class of subfunctors (intuitively,
predicates on computations) satisfying reasonable conditions,
including that it forms a factorization system on some monoidal
subcategory of the endofunctor category. Moreover, this grading is
canonical, namely universal in a certain 2-categorical sense. We also
show that algebraic operations of the given monad give rise to
\emph{flexibly graded} algebraic operations \cite{KMUW:flepgm} of the
canonically graded monad. Instead of working concretely with monads on a
category, we work abstractly with monoids in a (skew) monoidal category
equipped with a factorization system.

The structure of the paper is this. In \cref{sec:grading-objects}, we introduce
the idea of grading by subobjects for general objects and instantiate
this for grading of functors. We then proceed to gradings of monoids
and monads in \cref{sec:grading-monoids}. In \cref{sec:cartesian-grading}, we explore the
specific interesting case of grading monads canonically by subsets of
their sets of shapes. In \cref{sec:algebraic-operations}, we explain the emergence of
canonical flexibly graded algebraic operations for canonical gradings
of monads. One longer proof is in Appendix~\ref{sec:proofs}.

We introduce the necessary concepts regarding the classical topics of
monads, monoidal categories and factorization systems. For
additional background on
the more specific concepts of graded monad and skew monoidal category,
which we also introduce, we refer to \cite{katsumata2014,fujii2016towards} and \cite{Szl:skemcb,LS:triaos} as entry points.

\section{Grading objects and functors}\label{sec:grading-objects}
As a first step towards gradings of monoids, we introduce the notion of
a grading of an object of a category $\C$ with respect to a class of
morphisms $\M$ in $\C$.
We show that every object $T$ has a canonical such grading.
The case we care most about is when $\C$ is a category of endofunctors,
so that, in the next section, where we
extend these results to gradings of monoids, the monoids are exactly
monads.

\begin{definition}
  Let $\G$ be a category, whose objects $e$ we call \emph{grades}.
  A \emph{$\G$-graded object} of a category $\C$ is a
  functor $G : \G \to \C$.
\end{definition}
Let $\M$ be a class of morphisms of a category $\C$, and $T$ be an
object of $\C$.
There is a category $\slice{\M}{T}$, which has as objects
\emph{$\M$-subobjects} of $T$, i.e., pairs $(S, s)$
of an object $S$ and an $\M$-morphism $s : S \monoto T$.
Morphisms $f : (S, s) \to (S', s')$ are $\C$-morphisms $f : S \to S'$
such that $s = s' \compose f$.
We then have a $\slice{\M}{T}$-graded object $\slicedom{\M}{T}$ of $\C$,
defined by $\slicedom{\M}{T}(S, s) = S$.
This graded object forms an $\M$-grading in the sense of the following
definition, and is in fact the canonical $\M$-grading
(see \cref{thm:canonical-grading-object} below).

\begin{definition}
  Let $\M$ be a class of morphisms of a category $\C$.  An
  \emph{$\M$-grading} $(\G, G, g)$ of an object $T$ of $\C$ consists
  of a category $\G$, a functor $G : \G \to \C$ (= a $\G$-graded
  object of $\C$), and a natural transformation typed
  $g_d : Gd \monoto T$ whose components are all in $\M$.  A
  \emph{morphism} $(F, f) : (\G, G, g) \to (\G', G', g')$ between such
  gradings is a functor $F : \G \to \G'$ equipped with a natural
  isomorphism $f : G' \fcomp F \iso G$, such that $g_d \compose f_d = g'_{Fd}$.
  \[
    \begin{tikzcd}[row sep=small]
	\G && \TC \\
	\\
	{\G'} && \C
	\arrow["F"', from=1-1, to=3-1]
	\arrow["{G'}"', from=3-1, to=3-3]
	\arrow[""{name=0, anchor=center, inner sep=0}, "G"{description}, from=1-1, to=3-3]
	\arrow["T", from=1-3, to=3-3]
	\arrow[from=1-1, to=1-3]
	\arrow["f", shorten <=8pt, shorten >=8pt, Rightarrow, from=3-1, to=0]
	\arrow["g", shorten <=11pt, shorten >=11pt, Rightarrow, from=0, to=1-3]
\end{tikzcd}
\quad=\quad
  \begin{tikzcd}[row sep=small]
	\G && \TC \\
	\\
	{\G'} && \C
	\arrow["F"', from=1-1, to=3-1]
	\arrow["{G'}"', from=3-1, to=3-3]
	\arrow[""{name=0, anchor=center, inner sep=0}, "T", from=1-3, to=3-3]
	\arrow[from=1-1, to=1-3]
	\arrow[from=3-1, to=1-3]
        \arrow["{g'}"'{xshift=0.1em, yshift=0.5ex}, shorten <=28pt, shorten >=22pt, Rightarrow, from=3-1, to=0]
\end{tikzcd}
\]
  A \emph{2-cell} $\beta : (F, f) \tto (F', f')$ between such morphisms
  is a natural transformation $\beta : F \natto F'$ such that
  $f' \compose (G' \cdot \beta) = f$.
  These form a 2-category $\Grade{\M}{T}$.
\end{definition}

We organize gradings into a 2-category so that we can prove a universal
property (the following theorem) that characterizes the canonical
grading. The characterization is up to equivalence;
there is no reason to distinguish between isomorphic grades, therefore we
can work with gradings that are equivalent to the canonical one.
Working up to equivalence has the added benefit that, since
$\slice{\M}{T}$ is often equivalent to a small category, we often
have a canonical grading with a small set of grades.

\begin{theorem}\label{thm:canonical-grading-object}
  Let $\M$ be a class of a morphisms of a category $\C$, and let $T$ be an
  object of $\C$.  The data $(\slice{\M}{T}, \slicedom{\M}{T}, \slicemor{\M}{T})$,
  where $\slicedom{\M}{T} (S, s) = S$ and
  $\slicemor{\M}{T} (S, s) = s$, make a grading of $T$. This grading is
  canonical in the sense that it is the pseudoterminal object of
  $\Grade{\M}{T}$.
  Explicitly, for every other $\M$-grading
  $(\G, G, g)$ of $T$:
  \begin{itemize}
    \item there is a morphism
      $(F,f) : (\G,G,g) \to (\slice{\M}{T}, \slicedom{\M}{T},
      \slicemor{\M}{T})$ of $\M$-gradings;
    \item this morphism is essentially unique in the sense that
      there is a natural assignment of an isomorphism
      $(F', f') \iso (F, f)$ to every
      $(F', f') : (\G,G,g) \to (\slice{\M}{T}, \slicedom{\M}{T},
      \slicemor{\M}{T})$.
  \end{itemize}
\end{theorem}
\begin{proof}
  For existence, define
  $(F,f) : (\G,G,g) \to (\slice{\M}{T}, \slicedom{\M}{T}, \slicemor{\M}{T})$
  by
  \[
    Fd = (Gd, g_d) \quad Fh = Gh
    \qquad
    f_d = \id_{Gd}
  \]
  For uniqueness, given $(F', f')$, we have 2-cells
  $\beta_{(F', f')} : (F', f') \natto (F, f)$ and
  $\beta_{(F', f')}^{-1} : (F, f) \natto (F', f')$ given by $\beta_{(F',
  f'),d} = f'_d$ and
  $\beta^{-1}_{(F', f'), d} = f'^{-1}_d$.
  These are clearly natural in $(F', f')$ and inverse to each other,
  so we can use $\beta$ as the
  required natural isomorphism $(F', f') \iso (F, f)$.
\end{proof}

\begin{remark}
  In this paper, we discuss the problem of constructing canonical
  gradings, but one can also consider the dual problem of constructing
  canonical \emph{degradings}.
  For graded monads, this problem is discussed in
  \cite{fritz2018criterion,MPU:degl}.
  In the setting of this section, the initial degrading of a functor
  $G : \G \to \C$ would be the colimit $\colim G$, together with the
  morphisms $\inj{e} : Ge \to \colim G$ (when the colimit exists).
  The data $(\G, G, \inj{})$ is then an $\M$-grading of $\colim G$
  whenever $\inj{e}$ is in $\M$ for all $e$ (which is the case for our
  examples).
  This grading will typically not be the canonical grading of $\colim G$
  however.
  (For graded monads the situation is more complex: one does not take an
  ordinary colimit, but instead a colimit in a 2-category of monoidal
  categories, as discussed in \cite{fritz2018criterion}.)
\end{remark}

\subsection{Canonical gradings of endofunctors on \texorpdfstring{$\Set$}{Set}}
We give several examples for the case where $\C = [\Set, \Set]$ and $\M$
is the class of natural transformations whose components are injective functions.
In this case, every $\M$-subobject of an endofunctor $T : \Set \to \Set$
is isomorphic in $\slice{\M}{T}$ to a unique $\M$-subobject
$s : S \monoto T$ in which each injection $s_X$ is an inclusion.
In other words, $s$ is a choice of a subset $SX \subseteq TX$ for every
set $X$, closed under the action of $T$ in the sense that $x \in SX$
implies $Tfx \in SY$ for every $f : X \to Y$.
Below we characterize $\M/T$ up to equivalence for various
endofunctors $T$, using the fact that we need only consider the case
where $s$ is a family of inclusions.

\begin{example}
  The category $\M/\Id$ is equivalently the poset
  $\{\bot \leq \top\}$, with $\bot$ corresponding to the $\M$-subobject
  $S$ given by $SX = \emptyset$ for all $X$, and $\top$ corresponding to $SX = X$ for all $X$.
\end{example}

\begin{example}\label{ex:writer-functor}
  Consider the endofunctor $M \times T{-}$, where $M$ is a
  set and $T$ is an endofunctor on $\Set$.
  The category
  $
    \slice{\M}{(M \times T{-})}
  $,
  is equivalent to the category $(\slice{\M}{T})^M$, in which objects
  are $M$-indexed families $(\Sigma_z \in \slice{\M}{T})_{z \in M}$ of
  $\M$-subobjects of $T$.
  This is the case because, from every $S \monoto M \times T{-}$ in
  which each component is an inclusion, we can construct such a family
  $\mkSigma S$, and this construction forms a bijection with inverse
  $\mkS {{-}}$.
  \[
    {\mkSigma S}_z\mspace{1mu}X = \{x \in TX \mid (z, x) \in SX\}
    \qquad
    \mkS \Sigma X = \{(z, x) \in M \times TX \mid x \in \Sigma_z X\}
  \]
  In the special case $T = \Id$, we have
  $\slice{\M}{(M\times ({-}))} \equiv (\slice{\M}{\Id})^M \equiv
  \{\bot\leq\top\}^M \equiv (\mathcal P M, \subseteq)$, so
  the $\M$-subobjects of $M \times ({-})$ are equivalently the subsets
  of $M$, ordered by inclusion.
\end{example}

\begin{example}\label{ex:reader-functor}
  Consider the endofunctor $\exp{V}{({-})}$ (the underlying functor of
  the reader monad on $\Set$), where $V$ is a fixed set, and let $\M$ be
  the class of componentwise injective natural transformations.
  The $\M$-subobjects of $\exp{V}{({-})}$, and hence the objects of the
  canonical $\M$-grading of $\exp{V}{({-})}$, are equivalently
  upwards-closed sets of equivalence relations on $V$.

  To explain this in more detail, let $\Equiv V$ be the set of
  equivalence relations $R$ on $V$, considered as subsets
  $R \subseteq V \times V$.
  A function $f : V \to X$ \emph{respects} $R \in \Equiv V$ when
  $v\,R\,v'$ implies $f v = f v'$ for all $v, v' \in V$, equivalently,
  when $f$ factors through the quotient $[{-}]_R : V \to V/R$.
  A set $\Sigma \subseteq \Equiv V$ of equivalence relations is
  \emph{upwards-closed} when $R \in \Sigma$ implies $R' \in \Sigma$
  for all $R, R' \in \Equiv V$ with $R \subseteq R'$.
  Every such $\Sigma$ induces a subfunctor
  $\mkS \Sigma \monoto \exp{V}{({-})}$, defined by
  \[
    \mkS \Sigma\,X =
      \{f : V \to X \mid f~\text{respects some}~R \in \Sigma\}
  \]
  To go in the other direction, consider a subfunctor
  $S \monoto \exp{V}{({-})}$ in which every component of the
  $\M$-morphism is an inclusion.
  We obtain an upwards-closed $\mkSigma{S} \subseteq \Equiv V$:
  \[
    \mkSigma S = \{R \in \Equiv V \mid [{-}]_R \in S(V/R)\}
  \]
  This is upwards-closed because if $R \subseteq R'$ then
  $[{-}]_{R'} : V \to V/R'$ factors through $[{-}]_R : V \to V/R$, and
  since $S$ forms a functor, the family $S$ is closed under postcomposition.
  These two constructions are in bijection, with
  $\mkS{\mkSigma{S}} = S$ and $\mkSigma{\mkS{\Sigma}} = \Sigma$.
  It follows that $\M/{(\exp{V}{({-})})}$ is equivalent to the poset of
  upwards-closed sets $\Sigma \subseteq \Equiv{V}$, ordered by
  inclusion, and hence that this poset forms the canonical $\M$-grading of
  $\exp{V}{({-})}$.
\end{example}

\begin{example}\label{ex:state-functor}
  Consider the endofunctor $\exp{V}{V \times ({-})}$ (the underlying
  functor of the state monad), where $V$ is a set.
  Since $\exp{V}{V \times ({-})} \iso (\exp{V}{V}) \times (\exp{V}{({-})})$,
  we can combine \cref{ex:writer-functor,ex:reader-functor} to
  characterize the $\M$-subobjects of $\exp{V}{V \times ({-})}$.
  Every such subobject equivalently consists of an upwards-closed set
  $\Sigma_{p} \subseteq \Equiv{V}$ for each function $p : V \to V$.
  These can also be seen as subsets
  $\Sigma \subseteq (\exp{V}{V}) \times \Equiv{V}$
  such that $\{R \mid (p, R) \in \Sigma\}$ is upwards-closed for each
  $p : V \to V$.
  Given such a $\Sigma$, the corresponding $\M$-subobject
  $\mkS \Sigma \monoto (\exp{V}{V \times ({-})})$ is
  \[
    \mkS \Sigma X = \{f : V \to V \times X \mid
      \exists (p, R) \in \Sigma.~~
        \pi_1 \compose f = p~\wedge~\pi_2 \compose f \text{~respects~} R\}
  \]
\end{example}

\section{Grading monoids and monads}\label{sec:grading-monoids}

We proceed to grading monoids.
To define the notion of grading of a monoid, we need an appropriate
multiplication operation on the grades.
The obvious idea to to ask for the grades to form a monoidal category
instead of just a category, and much of the previous work on graded
monads (such as \cite{mellies2012parametric}) does exactly this.
However, in some of examples we do not get a monoidal category of
grades, but only a \emph{skew} monoidal category \cite{Szl:skemcb} of
grades.

\begin{definition}
  A \emph{(left-)skew monoidal} category is a category $\C$
  with a distinguished object
  $\I$, a functor $\ot : \C \times \C \to \C$ and three natural
  transformations $\lam$, $\rho$, $\al$ typed
\[
\lam_{X} : \I \ot X \to X \qquad
\rho_{X} : X \to X \ot \I \qquad
\al_{X,Y,Z} : (X \ot Y) \ot Z \to X \ot (Y \ot Z)
\]
satisfying the equations
\[
\small
\mathrm{(m1)}  
\begin{tikzcd}
    & \I \ot \I \arrow[dr, "\lambda_\I"] & \\
    \I \arrow[ur, "\rho_\I"] \ar[rr, equals]
    & & \I
\end{tikzcd}
\qquad
\mathrm{(m2)}\
\begin{tikzcd}
      (X\ot \I) \ot Y \arrow[r, "\alpha_{X,\I,Y}"]
      & X \ot (\I\ot Y) \arrow[d, "X\ot \lambda_{Y}"]\\
      X \ot Y \arrow[r, equals] \arrow[u, "\rho_{X}\ot Y"]
      &  X \ot Y 
\end{tikzcd}
\]
\[
\small
\mathrm{(m3)}\ 
\begin{tikzcd}
  (\I \ot X) \ot Y \arrow[rr, "\alpha_{\I,X,Y}"] \arrow[dr, "\lam_{X} \ot Y"']
  & &  \I \ot (X \ot Y) \arrow[dl, "\lambda_{X \ot Y}"] \\
  & X \ot Y
    & 
\end{tikzcd}
\qquad
\mathrm{(m4)}\ 
\begin{tikzcd}
  (X \ot Y) \ot \I \arrow[rr, "\alpha_{X,Y,\I}"] 
    & &  X \ot (Y \ot \I) \\
  & X \ot Y \arrow[ul, "\rho_{X \ot Y}"] \arrow[ur, "X \ot \rho_Y"'] 
\end{tikzcd}
\]
\[
\small
\mathrm{(m5)}\ 
\begin{tikzcd}
(X\ot (Y \ot Z)) \ot W \arrow[rr, "\alpha_{X,Y \ot Z,W}"]
  & & X\ot ((Y \ot Z)\ot W) \arrow[d, "X\ot \alpha_{Y,Z,W}"]
  \\
((X\ot Y) \ot Z) \ot W \arrow[u, "\alpha_{X,Y,Z} \ot W"]
      \arrow[r, "\alpha_{X\ot Y,Z,W}"]
  & (X\ot Y) \ot (Z \ot W) \arrow[r, "\alpha_{X,Y,Z\ot W}"]
    & X\ot (Y \ot (Z \ot W))
\end{tikzcd}
\]
$(\C, \I, \ot)$ is partially normal if one or several of $\lam$,
$\rho$ or $\alpha$ is a natural isomorphism. In particular, it is
\emph{left-normal} if $\lam$ is an isomorphism. A monoidal category is
a fully normal skew monoidal category.

A \emph{right-skew monoidal category} is given by
$(\C, \tensorI, \tensor, \lam, \rho, \al)$ such that the data
$(\C, \tensorI, \tensor^\mathrm{rev}, \rho, \lam, \al)$, where
$X \ot^\mathrm{rev} Y = Y \ot X$, form a left-skew monoidal category.
\end{definition}

\begin{definition}
  A \emph{monoid} in a skew monoidal category $(\C, \I, \ot)$
  is an object $T$ of $\C$ equipped with morphisms
\[
\eta : \I \to T \qquad \mu : T \ot T \to T
\]
satisfying the equations
\[
\begin{tikzcd}[row sep=small]
  T \arrow[ddrr, equals] \arrow[d, "\rho_T"']
  & \I \ot T \arrow[l, "\lambda_T"'] \arrow[r, "\eta \ot T"]
    & T \ot T \arrow[dd, "\mu"] \\
  T \ot \I \arrow[d, "T \ot \eta"']
  & & \\
  T \ot T \arrow[rr, "\mu"]
  & & T
\end{tikzcd} 
\qquad
\begin{tikzcd}[row sep=small]
  (T \ot T) \ot T \arrow[d, "\alpha_{T,T,T}"'] \arrow[rr, "\mu \ot T"]
  & & T \ot T \arrow[dd, "\mu"] \\
  T \ot (T \ot T) \arrow[d, "T \ot \mu"']
  & & \\
  T \ot T \arrow[rr, "\mu"]
  & & T
\end{tikzcd}  
\]  
\end{definition}

The concept of lax monoidal functor between skew monoidal categories
is defined as for monoidal categories; the same applies to the concept
of monoidal transformations between lax monoidal functors.

\begin{definition}
  Given a skew monoidal category $\monoidal{G} = (\G, \gI, \gtensor)$, a $\monoidal{G}$-\emph{graded monoid} in a skew monoidal category
  $\monoidal{C} = (\C, \tensorI, \tensor)$ is the same as a lax monoidal functor
  $\lax{G} : \monoidal{G} \to \monoidal{C}$. Explicitly, it is
  a functor $G : \G \to \C$ with a morphism $\eta : \tensorI \to G \gI$
  and a natural transformation typed
  $\mu_{d,d'} : G d \tensor G d' \to G (d \gtensor d')$ subject to
  equations similar to those of a monoid.
\end{definition}

\begin{definition}
  Let $\T = (T, \eta, \mu)$ be a monoid in a skew monoidal category
  $\monoidal{C} = (\C, \tensorI, \tensor)$, and let $\M$ be a class of morphisms of
  $\C$.
  An $\M$-\emph{grading} $(\monoidal{G}, \lax{G}, g)$ of the monoid $\T$ consists
  of a skew monoidal category $\monoidal{G}$, a lax monoidal functor
  $\lax{G} : \monoidal{G} \to \monoidal{C}$ (= a $\monoidal{G}$-graded monoid in $\monoidal{C}$), and a monoidal transformation typed $g_d : Gd \monoto T$,
  whose components are all in $\M$.
  A \emph{morphism}
  $(\lax{F}, f) : (\monoidal{G},\lax{G}, g) \to (\monoidal{G}',\lax{G}', g')$
  between such gradings is a lax monoidal functor
  $\lax{F} : \monoidal{G} \to \monoidal{G'}$ equipped with a monoidal
  isomorphism $f : \lax{G}' \fcomp \lax{F} \iso \lax{G}$,
  such that $g_d \compose f_d = g'_{Fd}$.
  A \emph{2-cell} $\beta : (\lax{F}, f) \tto (\lax{F}', f')$ is a monoidal
  transformation $\beta : \lax{F} \natto \lax{F}'$ such that
  $f' \compose (\lax{G}' \fcomp \beta) = f$.
  We write $\Grade {\M} {\T}$ for this 2-category.
\end{definition}

\begin{example}\label{ex:state-grading}
  The situation we are mainly interested in is when
  $\C = [\D, \D]$ is the category of endofunctors on some $\D$,
  with the identity for $\tensorI$ and functor composition for $\tensor$.
  In this case, monoids in $\monoidal{C}$ are exactly monads on $\D$, and
  a lax monoidal functor
  $\lax{G} : \monoidal{G} \to \monoidal{C}$ (a $\monoidal{G}$-graded monoid in $\monoidal{C}$) is a
  $\monoidal{G}$-graded monad on $\D$, in the sense of
  \cite{Smirnov2008,mellies2012parametric,katsumata2014}.
  Explicitly, the unit and multiplication of $\lax{G}$ have the form
  \[
    \eta_X : X \to G \gI X
    \qquad
    \mu_{e, e', X} : G e (G e' X) \to G (e\gtensor e') X
  \]

  For a concrete example, let $V$ be a set (of states), and let $\T$ be
  the state monad over $V$:
  \[
    T X = \exp{V}{V \times X}
    \qquad
    \eta_X\,x\,v = (v, x)
    \qquad
    \mu_X\,f\,v = g\,v'~~\text{where}~(v', g) = f\,v
  \]
  We give a $\M$-grading $(\monoidal{G}, \lax{G}, g)$ of $\T$, where
  $\M$ is componentwise injective natural transformations.
  Let $\G$ be the poset of subsets of $\{\mathsf{get}, \mathsf{put}\}$
  ordered by inclusion, which forms a strict monoidal category with
  $\emptyset$ for the unit $\gI$ and $e \cup e'$ for the tensor $e \gtensor e'$.
  We then define $\lax{G}$ by
  \[
    \begin{array}{r@{}l}
      G e = \{ f : V&\,\to V \times X
      \\\mid~&
        \mathsf{get} \not\in e \Rightarrow (
          \pi_1 \compose f~\text{is a constant function or~$\id_V$}
          ~\wedge~\pi_2 \compose f~\text{is a constant function})
      \\&
        \wedge~\mathsf{put} \not\in e \Rightarrow \pi_1 \compose
        f~\text{is $\id_V$} \}
    \end{array}
  \]
  with unit and multiplication defined as for $\T$.
  This forms an $\M$-grading with the inclusions for $g$.

  This grading is suitable for interpreting a Gifford-style effect
  system~\cite{lucassen1988polymorphic} for global state.
  A function $f : V \to V \times X$ is a computation
  $f \in TX$, sending an initial state to a pair of a final state and
  a result.
  A grade $e$ gives the set of operations that a computation may use
  when it is executed, so $GeX$ is the subset of $TX$ on the
  computation that only use the operations in $e$.
  For example, $G\{\mathsf{get}\}X$ contains computations that may use the
  initial state, but do not change the state (with $\mathsf{put}$).
\end{example}

We turn now to the \emph{canonical} grading of a monoid $\T$ in a
skew monoidal category $\monoidal{C}$.
Since the category $\M/T$ forms the canonical grading of the
\emph{object} $T$, we show that (under the conditions explained below),
we can make $\M/T$ into a skew monoidal category, using the monoid
structure of $\T$.

First consider the slice category $\slice{\C}{T}$ (where we do not
restrict to $\M$-morphisms).
This already forms a skew monoidal category, with
\[
  \I \xto{\eta} T
  \qquad
  S \tensor S' \xto{s \tensor s'} T \tensor T \xto{\mu} T
\]
for the unit and tensor of $(S, s)$ and $(S', s')$ (see for example
Kelly~\cite{kelly1992clubs}).
In general this skew monoidal structure will not restrict to $\M/T$,
because the morphism $\mu$ is not in $\M$ for many of our examples.
However, we can make $\M/T$ into a skew monoidal category by
adapting the skew monoidal structure on $\C/T$.
The idea is to just factorize the morphisms we use in the tensor and
unit of $\C/T$ to obtain morphisms in $\M$.
Hence we ask that $\M$ forms an orthogonal factorization system in the
usual sense.

\begin{definition}
  An (orthogonal) \emph{factorization system} on a category $\C$ is a
  pair $(\E, \M)$ of classes of morphisms of $\C$, such that
  \begin{itemize}
    \item both $\E$ and $\M$ contain all isomorphisms, and are closed
      under composition;
    \item $\E$-morphisms are \emph{orthogonal} to $\M$-morphisms:
      for every commuting square in $\C$ as on the left below, with $e \in \E$
      and $m \in \M$, there is a unique $d$ making the diagram
      on the right below commute.
      \[
        \begin{tikzcd}[row sep=small]
          X \arrow[r, "e", two heads] \arrow[d, "f"'] &
          Y \arrow[d, "g"] \\
          X' \arrow[r, "m"', tail] &
          Y'
        \end{tikzcd}
        \qquad\qquad
        \begin{tikzcd}[row sep=small]
          X \arrow[r, "e", two heads] \arrow[d, "f"'] &
          Y \arrow[d, "g"] \arrow[ld, "d" description, dashed] \\
          X' \arrow[r, "m"', tail] &
          Y'
        \end{tikzcd}
      \]
    \item every morphism $f : X \to Y$ in $\C$ has an $(\E,
      \M)$-factorization: there exist an
      object $S$, an $\E$-morphism $e : X \epito S$, and an $\M$-morphism
      $m : S \monoto Y$ such that the diagram below commutes.
      \[
        \begin{tikzcd}[row sep=tiny]
          X \arrow[rr, "f"] \arrow[rd, "e"', two heads] &&
          Y \\
          &
          S \arrow[ru, "m"', tail]
        \end{tikzcd}
      \]
  \end{itemize}
\end{definition}

\begin{example}
  Writing $\Mor$ for the class of all morphisms and $\Iso$ for the class
  of isomorphisms, every category has
  $(\Mor, \Iso)$ and $(\Iso, \Mor)$ as factorization systems.
  On $\Set$, the classes of surjective and of injective functions form a
  factorization system $(\Surj, \Inj)$.
  To factorize a function $f : X \to Y$ in this case, we let
  $S = \{f\,x \mid x \in X\}$ be the image of $f$, and define
  \begin{tikzcd}[cramped,column sep=1.6em]
    X \arrow[r, "e", two heads] &
    S \arrow[r, "m", tail] &
    Y
  \end{tikzcd}
  by $e\,x = f\,x$ and $m\,y = y$.
  On the category $\Poset$ of partially ordered sets and monotone
  functions, we have a factorization system $(\Surj, \Full)$, where
  $\Surj$ is the class of surjective monotone functions and $\Full$ is
  the class of full functions, i.e.\ monotone functions $m : S \to Y$
  such that $m\,x \leq m\,y$ implies $x \leq y$.
  Factorizations are given as in $\Set$, with the order on $S$ inherited
  from the order on $Y$.
\end{example}

We are primarily interested in canonically grading monads, which are
monoids in the endofunctor category $\C = [\D, \D]$, with functor
composition as the tensor.
We therefore want a factorization system on $[\D, \D]$.
We give the most standard option for this as the following example,
but there are others we are interested in (see \cref{lemma:partly-cartesian-factorizations} below).
In models of computational effects we in fact usually want a
\emph{strong} monad.
If $\D$ is monoidal, then a \emph{strong endofunctor} on $\D$ is a
functor $F : \D \to \D$ equipped with a \emph{strength}, i.e.\ a
natural transformation
$\str_{\Gamma, X} : \Gamma \tensor_{\D} FX \to F(\Gamma \tensor_{\D} X)$
satisfying two laws for compatibility with the left unitor and
associator of $\D$.
These form a monoidal category $[\D, \D]_s$, in which morphisms are
strength-preserving natural transformations and the tensor is
composition.
Strong monads are monoids in $[\D, \D]_s$.
Below we consider non-strong monads for simplicity, but we can also
apply our results to strong monads using a factorization system on
$[\D, \D]_s$.

\begin{example}
  If $(\E, \M)$ is a factorization system on a category $\D$, then
  the endofunctor category $[\D, \D]$ has a factorization system
  $(\text{componentwise-}\E, \text{componentwise-}\M)$.
  Factorizations $F \epito S \monoto G$ of natural transformations are
  componentwise.
  If $\D$ is monoidal and $\E$ is closed under $\Gamma \tensor_\D ({-})$
  for all $\Gamma$ then
  $(\text{componentwise-}\E, \text{componentwise-}\M)$ is a
  factorization system on $[\D, \D]_s$.
  Morphisms are again factorized componentwise; for the construction of
  the strength for $S$ see \cite[Section~2.2]{kammar2018factorisation}.
\end{example}

Forming a factorization system $(\E, \M)$ is merely a property of a class
$\M$ of morphisms, because $\E$ is necessarily the class of all morphisms $e$ that
are orthogonal to all $\M$-morphisms.
Factorizations of morphisms are unique up to unique isomorphism.

If $\M$ forms a factorization system $(\E, \M)$, then for a given monoid
$\T$ we construct a unit $\cI$ and a tensor $\ctensor$ for the
category $\slice{\M}{T}$  by factorizing morphisms as follows:
\[
  \begin{tikzcd}[row sep=tiny]
    \tensorI \arrow[rr, "\eta"] \arrow[rd, "q"', two heads] &&
    T \\
    &
    \cI \arrow[ru, tail]
  \end{tikzcd}
  \qquad
  \begin{tikzcd}[row sep=tiny]
    S \tensor S'
    \arrow[r, "s \tensor s'"]
    \arrow[rd, "q_{S, S'}"', two heads] &
    T \tensor T
    \arrow[r, "\mu"] &
    T \\
    &
    S \ctensor S'
    \arrow[ru, tail]
  \end{tikzcd}
\]
To construct the required structural morphisms,
we make the additional assumption that $\E$ is closed under
$({-}) \tensor S$ for every $S \monoto T$.
Under this assumption, the following squares, which all commute, have a unique
diagonal, and these diagonals are the required structural morphisms.
\[
  \begin{tikzcd}[column sep=large, row sep=small]
    S_1 \tensor S'_1
    \arrow[r, "q_{S_1, S'_1}", two heads]
    \arrow[d, "f \tensor f'"'] &
    S_1 \ctensor S'_1
    \arrow[dd, tail]
    \arrow[ldd, "f \ctensor f'"{description}, dashed] \\
    S_2 \tensor S'_2
    \arrow[d, "q_{S_2, S'_2}"', two heads] \\
    S_2 \ctensor S'_2
    \arrow[r, tail] &
    T
  \end{tikzcd}
  ~~\text{where}~~
  \begin{tikzcd}[column sep=0.5em, row sep=huge]
    S_1 \arrow[rr, "f"] \arrow[rd, tail] &&
    S_2 \arrow[ld, tail] \\
    &
    T
  \end{tikzcd}
  \begin{tikzcd}[column sep=0.5em, row sep=huge]
    S'_1 \arrow[rr, "f'"] \arrow[rd, tail] &&
    S'_2 \arrow[ld, tail] \\
    &
    T
  \end{tikzcd}
  ~~\text{are morphisms in $\M/T$}
\]
\[
  \begin{tikzcd}[column sep=small, row sep=large]
    \tensorI \tensor S
    \arrow[r, "q \tensor S", two heads]
    \arrow[d, "\lambda_S"'] &
    \cI \tensor S
    \arrow[r, "q_{\cI, S}", two heads] &
    \cI \ctensor S
    \arrow[d, tail]
    \arrow[lld, "\ell_S"{description}, dashed] \\
    S \arrow[rr, tail] &&
    T
  \end{tikzcd}
  \quad
  \begin{tikzcd}[row sep=small, column sep=normal]
    S
    \arrow[r, equals]
    \arrow[d, "\rho_S"'] &
    S
    \arrow[ddd, tail]
    \arrow[lddd, "r_S"{description}, dashed] \\
    S \tensor \tensorI
    \arrow[d, "S \tensor q"'] \\
    S \tensor \cI
    \arrow[d, "q_{S, \cI}"', two heads] \\
    S \ctensor \cI
    \arrow[r, tail] &
    T
  \end{tikzcd}
  \quad
  \begin{tikzcd}[sep=1.8em]
    (S \tensor S') \tensor S''
    \arrow[r, "q_{S, S'} \tensor S''", two heads]
    \arrow[d, "\alpha_{S, S', S''}"'] &
    (S \ctensor S') \tensor S''
    \arrow[r, "q_{S \ctensor S', S''}", two heads] &
    (S \ctensor S') \ctensor S''
    \arrow[ddd, tail]
    \arrow[llddd, "a_{S, S', S''}"{description}, dashed] \\
    S \tensor (S' \tensor S'')
    \arrow[d, "S \tensor q_{S', S''}"'] \\
    S \tensor (S' \ctensor S'')
    \arrow[d, "q_{S, S' \ctensor S''}"', two heads] \\
    S \ctensor (S' \ctensor S'')
    \arrow[rr, tail] &&
    T
  \end{tikzcd}
\]

If $\lambda$ is a natural isomorphism, then so is $\ell$. This does
not apply to $\rho$ and $\alpha$ unless $\E$ is also closed under
$S \tensor (-)$ for all $S \monoto T$.

\begin{theorem}\label{thm:slice-mon}
  Let $\T$ be a monoid in a skew monoidal category
  $\monoidal{C} = (\C, \tensorI, \tensor)$, and let $(\E, \M)$ be a
  factorization system on $\C$.  If $\E$ is closed under
  $({-}) \tensor S$ for every $\M$-subobject $S \monoto T$, then
  $\slice{\M}{\T} = (\slice{\M}{T}, \cI, \ctensor)$ is a skew monoidal
  category. If $\monoidal{C}$ is left-normal, then
  so is $\slice{\M}{\T}$.  If $\E$ is also closed
  under $S \tensor ({-})$ for every $S \monoto T$ and
  $\monoidal{C}$  is monoidal, then
  $\slice{\M}{\T}$ is monoidal.
\end{theorem}

\begin{proof}
  See Appendix~\ref{sec:proofs}.
\end{proof}  

\begin{remark}
  Closure of $\E$ under $T \tensor ({-})$ does not in general imply
  closure of $\E$ under $S \tensor ({-})$ for $S \monoto T$.
  Consider the factorization system
  $(\E, \M) = (\text{componentwise surjective}, \text{componentwise full})$ on $[\Poset, \Poset]$, with
  composition as $\tensor$.
  For an endofunctor $F : \Poset \to \Poset$, closure of $\E$
  under $F \tensor ({-})$ amounts to closure of $\Surj$ under $F$.
  The class $\Surj$ is closed under $\exp{V}{({-})}$ exactly when
  $V$ is discrete.
  Hence, while this property holds for $\exp{\{0,1\}}{({-})}$, it does
  not hold for $\exp{\{0 \leq 1\}}{({-})} \monoto \exp{\{0,1\}}{({-})}$.

  There are cases in which $\E$ is closed under $S \tensor {({-})}$ for
  all $S \monoto T$ even if $\E$ is not closed under
  $F \tensor {({-})}$ for general $F$:
  for example every $\M$-subobject of the endofunctor $M \times ({-})$
  on $\Poset$ has
  the form $S \times ({-})$ for some $S \monoto M$, and functors of the
  form $S \times ({-})$ send surjections to surjections.
\end{remark}

Our task is now to show that the skew monoidal category
$\M/\T = (\M/T, \cI, \ctensor)$ forms the canonical grading
$(\M/\T, \lax{T}_{\M}, \slicemor{\M}{\T})$ of the monoid $\T$ when $\E$ is closed
under $S \tensor ({-})$ for each $S \monoto T$.
The lax monoidal functor
$\lax{T}_{\M} : \slice{\M}{\T} \to \monoidal{C}$ is given on
objects by
\[
  \slicedom{\M}{T} (S, s : S \monoto T) = S
\]
and has as unit and multiplication the $\E$-morphisms from the
construction of $\cI$ and $\ctensor$:
\[
  \begin{tikzcd}
    \tensorI \arrow[r, "q", two heads] & \slicedom{\M}{T} \cI
  \end{tikzcd}
  \qquad
  \begin{tikzcd}
    \slicedom{\M}{T} (S, s) \tensor \slicedom{\M}{T} (S', s')
    \arrow[r, "q_{S, S'}", two heads] &
    \slicedom{\M}{T} ((S, s) \ctensor (S', s'))
  \end{tikzcd}
\]
That this is lax monoidal is immediate from the definition of the
structural morphisms of $\slice{\M}{\T}$.
Finally, the monoidal transformation
$\slicemor{\M}{\T} : \slicedom{\M}{\T} \tto \T$ is given by
$\slicemor{\M}{\T}_{(S, s)} = s$.
Monoidality of $\slicemor{\M}{\T}$ is immediate from the definitions of
$\cI$ and $\ctensor$.
Hence $(\slice{\M}{\T}, \slicedom{\M}{\T}, \slicemor{\M}{\T})$ is a grading of $\T$.
Canonicity is the following theorem.

\begin{theorem}\label{thm:canonical-grading}
  Let $\T$ be a monoid in a skew monoidal category $\monoidal{C} = (\C, \tensorI, \tensor)$,
  and let $(\E, \M)$ be a factorization system on $\C$ such that
  $\E$ is closed under $({-}) \tensor S$ for each $\M$-subobject $S$ of
  $T$.
  The grading $(\slice{\M}{\T}, \slicedom{\M}{\T}, \slicemor{\M}{\T})$ is canonical
  in the sense that it is the pseudoterminal object of
  $\Grade{\M}{\T}$.
  Explicitly, for every $\M$-grading $(\monoidal{G}, \lax{G}, g)$ of the
  monoid $\T$:
  \begin{itemize}
    \item there is a morphism $(\lax{F}, f) : (\monoidal{G}, \lax{G}, g) \to (\slice{\M}{\T}, \slicedom{\M}{\T}, \slicemor{\M}{\T})$,
  of $\M$-gradings of $\T$;
    \item this morphism is essentially unique in the sense that
      there is a natural assignment of an isomorphism
      $(\lax{F'}, f') \iso (F, f)$ to every
      $(\lax{F'}, f') : (\monoidal{G},\lax{G},g) \to (\slice{\M}{\T}, \slicedom{\M}{\T}, \slicemor{\M}{\T})$.
    \end{itemize}
\end{theorem}
\begin{proof}
  We have done most of the proof already as
  \cref{thm:canonical-grading-object}; we fill in the remaining parts.
  Recall from there that we define
  \[
    Fd = (Gd, g_d) \quad Fh = Gh
    \qquad
    f_d = \id_{Gd}
  \]
  We make $F$ into a lax monoidal functor by using the unique diagonals
  of the following squares as the unit and multiplication. The squares
  commute because $\lax{G}$ is lax monoidal.
  \[
    \begin{tikzcd}[column sep=large]
      \tensorI
      \arrow[r, "q", two heads]
      \arrow[d, "\eta"'] &
      \cI
      \arrow[d, tail]
      \arrow[ld, "\eta"{description}, dashed] \\
      G \gI
      \arrow[r, "g_{\gI}"', tail] &
      T
    \end{tikzcd}
    \qquad
    \begin{tikzcd}
      G d \tensor Gd'
      \arrow[r, "q_{Gd, Gd'}", two heads]
      \arrow[d, "\mu_{d, d'}"'] &
      G d \ctensor Gd'
      \arrow[d, tail]
      \arrow[ld, "\mu_{d,d'}"{description}, dashed] \\
      G(d \gtensor d')
      \arrow[r, "g_{d \gtensor d'}"', tail] &
      T
    \end{tikzcd}
  \]
  This definition immediately implies that $f$ is monoidal, and hence
  that $(\lax F, f)$ is a morphism of $\M$-gradings of $\T$.
  Finally, given $(\lax{F'}, f')$, we show that $\beta_d = f'_d$ defines an
  isomorphism $\beta : (\lax{F'}, f') \iso (\lax{F}, f)$.
  For this it remains to show that $\beta$ is monoidal (it follows
  automatically that the inverse $\beta^{-1}$ is monoidal).
  For compatibility with the multiplications, this amounts to showing
  that the square on the left below commutes.
  For this it is enough to show both paths in that square provide
  the unique diagonal of the square on the right, and this follows from
  the fact that $f'$ is monoidal.
  \[
    \begin{tikzcd}[row sep=large]
      F' d \ctensor F' d'
      \arrow[r, "f'_d \ctensor f'_{d'}"]
      \arrow[d, "\mu_{d, d'}"'] &
      G d \ctensor G d'
      \arrow[d, "\mu_{d, d'}"] \\
      F' (d \gtensor d')
      \arrow[r, "f'_{d \gtensor d'}"'] &
      G (d \gtensor d')
    \end{tikzcd}
    \qquad
    \begin{tikzcd}[row sep=small]
      F' d \tensor F' d'
      \arrow[r, "q_{F'd, F'd'}", two heads]
      \arrow[d, "f'_d \tensor f'_{d'}"'] &
      F' d \ctensor F' d'
      \arrow[dd, tail]
      \arrow[ldd, dashed] \\
      G d \tensor Gd'
      \arrow[d, "\mu_{d, d'}"'] \\
      G (d \gtensor d')
      \arrow[r, "g_{d \ctensor d'}"', tail] &
      T
    \end{tikzcd}
  \]
  Compatibility with the units is similar.
\end{proof}

\begin{example}\label{ex:writer-monoidal}
  Let $(M, \eps, \cdot)$ be a monoid in the cartesian monoidal
  category $\Set$, and let $\T$ be the corresponding
  writer monad, which has endofunctor
  $
    TX = M \times X
  $,
  unit
  $
    \eta_X x = (\eps, x)
  $,
  and multiplication
  $
    \mu_X (z, (z', x)) = (z \cdot z', x)
  $.
  Then $\T$ is a monoid in the monoidal category of endofunctors on
  $\Set$, with functor composition.
  Consider the factorization system $(\E, \M)$ on $[\Set, \Set]$ in
  which $\E$ (respectively $\M$) is componentwise surjective (resp.\
  injective) natural transformations.
  The class $\E$ is closed under functor composition on both sides, so
  $\slice{\M}{\T}$ is monoidal and provides the canonical grading of
  $\T$.
  We show in \cref{ex:writer-functor} that $\M$-subobjects of
  $T$ are equivalently subsets $\Sigma \subseteq M$.
  Under this equivalence, the monoidal structure on
  $\slice{\M}{\T}$ is given by
  $
    \cI = \{\eps\}
  $
  and
  $
    \Sigma \ctensor \Sigma'
      = \{z \cdot z' \mid z \in \Sigma, z' \in \Sigma'\}
  $.
  The graded monad $\slicedom{\M}{\T}$ is given by
  \[
    \slicedom{\M}{T} \Sigma = \Sigma \times X
    \qquad
    \eta_X x = (\eps, x)
    \qquad
    \mu_{\Sigma, \Sigma', X} (z, (z', x)) = (z \cdot z', x)
  \]
\end{example}

\begin{example}\label{ex:state-monoidal}
  We show in \cref{ex:state-functor} that, when $\M$ is componentwise
  injective natural transformations and $T = \exp{V}{V \times ({-})}$, the
  objects of $\slice{\M}{T}$ are equivalently subsets
  $\Sigma \subseteq (\exp{V}{V}) \times \Equiv V$ satisfying a closure
  condition.
  When $\T$ is the state monad, these form a monoidal category $\M/\T$,
  and the graded monad $\slicedom{\M}{\T}$ has underlying functor
  \[
    \slicedom{\M}{T} \Sigma X = \{ f : V \to V \times X \mid
      \exists (p, R) \in \Sigma.\,
        p = \pi_1 \compose f \wedge
        (\pi_2 \compose f)~\text{respects}~R\}
  \]
  \Cref{ex:state-grading} provides another grading of $\T$, in which the
  grades are subsets of $\{\mathsf{get}, \mathsf{put}\}$.
  By \cref{thm:canonical-grading}, we obtain a morphism $(\lax{F}, f)$ of
  gradings.
  Under the characterization of grades as subsets $\Sigma$, the
  underlying functor $F$ sends $e \subseteq \{\mathsf{get},
  \mathsf{put}\}$ to $Fe \subseteq (\exp{V}{V}) \times \Equiv V$ as follows:
  \begin{gather*}
    F\emptyset = \{(\id_V, V \times V)\}
    \qquad
    F\{\mathsf{get}, \mathsf{put}\} = (\exp{V}{V}) \times \Equiv V
    \\
    F\{\mathsf{get}\} = \{(\id_V, R) \mid R \in \Equiv V\}
    \qquad
    F\{\mathsf{put}\} = \{
      (p, V \times V) \mid p~\text{is a constant function or}~\id_V\}
  \end{gather*}
\end{example}

\section{Canonical grading by sets of shapes}
\label{sec:cartesian-grading}
When assigning grades to computations $t \in TX$, where $\T$ is a monad
on $\Set$, we are often interested only in the \emph{shape} of the
computation.
A \emph{shape} is an element of $T1$, where $1$ is the one-element set;
and the shape of the computation $t \in TX$ is $T !\,t \in T1$, where
$!$ is the unique function $X \to 1$.
A grade in this case is a subset $e \subseteq T1$ of the set of
shapes, and a computation has grade $e$ when its shape $T!\,t$ is in
$e$.

More generally, if $\T$ is a monad on a category $\D$ with a terminal
object $1$, then the object of shapes is $T1$.
Given a class $\M$ of morphisms of $\D$, we can consider grading by
$\M$-subobjects of $T1$.
We show in this section that these grades can be considered canonical,
using a suitable class $\M'$ of morphisms of $[\D, \D]$.

\begin{definition}
  Let $\A$ be a category.
  A natural transformation $f : F \natto G : \A \to \D$ is
  \emph{cartesian} if all of its naturality squares are pullbacks.
  If $\A$ has pullbacks, then a functor $F : \A \to \D$ is
  \emph{cartesian} when it preserves pullbacks.
\end{definition}

\begin{lemma}\label{lemma:partly-cartesian-factorizations}
  Let $(\E, \M)$ be a factorization system on a category $\D$ with
  pullbacks, and let $\A$ be a category with a terminal object.
  Then we have a factorization system $(\E', \M')$ on $[\A, \D]$ as
  follows:
  \begin{align*}
    \E' &= \textup{natural transformations}~e
      ~\textup{such that}~e_1 \in \E
    \\
    \M' &= \textup{cartesian natural transformations}~m
      ~\textup{such that}~m_1 \in \M
  \end{align*}
  For every functor $G : \A \to \D$, there is an equivalence of
  categories $\slice{\M'}{G} \equiv \slice{\M}{G1}$.
\end{lemma}
Before giving the proof, we note that Kelly~\cite{kelly1992clubs}
considers $(\E', \M')$ in the case $(\E, \M) = (\Iso, \Mor)$.
\begin{proof}
  That $\E'$ and $\M'$ are closed under isomorphisms and composition is
  straightforward.
  To factorize a natural transformation $\tau : F \natto G$, we first
  factorize $\tau_1$ using $(\E,\M)$, as on the bottom of the following
  diagram.
\[
\begin{tikzcd}
  F X \arrow[d, "F!"']
       \arrow[r, "e_X", dashed] \arrow[rr, bend left=30, "\tau_X"] 
  & S X \arrow[r, "m_X"] \arrow[d]
    & G X \arrow[d, "G!"]  \\
  F 1 \arrow[r, "\underline{e}", two heads] \arrow[rr, bend right=30, "\tau_1"'] 
  & \underline{S} \arrow[r, "\underline{m}", tail]
    & G 1  
\end{tikzcd}  
\]
  We then factorize any component $\tau_X$ as on the top, by taking as
  $(SX, m_X)$ the pullback of $(\underline{S},\underline{m})$ along
  $G!$, and taking as $e_X$ the unique map from $F X$ to this pullback.
  The objects $SX$ form a functor using unique maps into pullbacks, and
  the morphisms $e_X$ and $m_X$ are natural transformations.
  When $X = 1$ we have $e_X \in \E$ and $m_X \in \M$ because the
  vertical morphisms in the diagram are all isomorphisms.
  Hence we have the required factorization of $\tau$ into $e \in \E'$ and
  $m \in \M'$.
  For orthogonality, unique diagonal fill-ins are constructed as unique
  maps to pullbacks.

The required equivalence of categories exists because every
$\M'$-subobject $m : S \monoto G$ is determined up to
isomorphism by the component $m_1$.
The latter is the corresponding object of $\slice{\M}{G1}$.
\end{proof}

\cref{lemma:partly-cartesian-factorizations} provides a construction of a factorization system
$(\E', \M')$ in particular on endofunctor categories $[\D, \D]$ when
$\D$ has pullbacks and a terminal object.
In this case the $\M'$-subobjects of $T$ are $\M$-subobjects of $T1$.
However, $\E'$ is often closed under neither $(-) \fcomp S$ nor $S \fcomp (-)$ for $S \monoto T$, and $\slice{\M'}{T}$ is neither left-skew nor right-skew monoidal
for a monad $\T$.
In the following example, left-skew monoidality fails, but we do get right-skew monoidality. 

\begin{example}
  Consider the factorization system $(\E, \M) = (\Surj, \Inj)$ on $\A = \D = \Set$. Surjections in $\Set$ are preserved by any functor $S$, therefore the class $\E'$ of the factorization system $(\E', \M')$ on $[\Set,\Set]$ is closed under $S \fcomp (-)$ for any $S$ (as $(S \fcomp e)_1 = S e_1$). Given a set monad $\T$, the category $\slice{\M'}{T}$ obtains a right-skew monoidal structure by the ``reversal'' of \cref{thm:slice-mon}.

  Let $\T$ be the state monad for a set of states $V$:
  \[
    T = \exp{V}{V \times ({-})}
    \qquad
    \eta_X\,x\,v = (v, x)
    \qquad
    \mu_X\,f\,v = g\,v'~~\text{where}~(v', g) = f\,v
  \]
  Since $T1 \iso \exp{V}{V}$, we have
  $\slice{\M'}{T} \equiv \slice{\Inj}{(\exp{V}{V})}$, so that the
  canonical grades are equivalently subsets of
  $\Sigma \subseteq \exp{V}{V}$, ordered by inclusion.
  A subset $\Sigma$ corresponds to the $\M'$-subobject
  $\mkS{\Sigma} \monoto T$ given by
  $
    \mkS{\Sigma} X
      = \{f : V \to V \times X \mid \pi_1 \compose f \in \Sigma\}
  $.
  Given such a subset, define
  $
    \mathrm{Cl}(\Sigma')
      = \{f : V \to V \times X \mid \forall v.\,\exists g \in \Sigma'.\,
            \pi_1 (f\,v) = g\,v\}
  $.
  The right-skew monoidal category of canonical grades $\Sigma$ has 
  unit $\cI = \{\id_V\}$ and tensor
  $
    \Sigma \ctensor \Sigma' = \Sigma \mathbin{\hat\compose} \mathrm{Cl}(\Sigma')
  $,
  where
  $
    \Sigma \mathbin{\hat\compose} \Sigma'
      = \{f' \compose f \mid f \in \Sigma, f' \in \Sigma'\}
  $.
  There is no left unitor for a left-skew monoidal structure, because
  $\cI \ctensor \Sigma' = \mathrm{Cl}(\Sigma')$ is not in general equal
  to $\Sigma'$.
\end{example}

The failure of left-skew monoidality in this example can be traced back to the
failure of $\E'$ to be closed under $({-}) \fcomp S$ for 
endofunctors $S \monoto T$.
We had no such problem for the componentwise lifting of
$(\E, \M)$.
When $(\E, \M)$ is \emph{stable} (\cref{def:stable} below) and one restricts
$[\A,\D]$ to cartesian natural transformations (and optionally further
also to cartesian functors), then the componentwise lifting actually
coincides with $(\E', \M')$, as we show in the next few lemmata.
This then provides sufficient conditions for
$\E'$ to be closed under $({-}) \fcomp S$.
We give an example in which these conditions are satisfied in
\cref{ex:list-cartesian} below, where we actually obtain a monoidal
structure on the category of grades.

\begin{lemma}\label{lemma:cartesian-subfunctors}
  If $\A$ has pullbacks, $m : S \natto G : \A \to \D$ is a cartesian
  natural transformation, and $G$ is cartesian, then $S$ is also
  cartesian.
\end{lemma}
\begin{proof}
  Every pullback square in $\A$, as on the left below, induces a cube in
  $\D$, as on the right below.
  Four of the faces of this cube are pullbacks because $m$
  is cartesian, and the face on the right is a pullback because $G$ is
  cartesian.  It follows that the left face is also a pullback.
\[
  \begin{tikzcd}[remember picture]
    X & {X'} \\
    Y & {Y'}
    \arrow["x", from=1-1, to=1-2]
    \arrow["{f'}", from=1-2, to=2-2]
    \arrow["f"', from=1-1, to=2-1]
    \arrow["y"', from=2-1, to=2-2]
  \end{tikzcd}
  \begin{tikzpicture}[overlay, remember picture]
    \coordinate (NW) at (\tikzcdmatrixname-1-1);
    \coordinate (SW) at ([xshift=0.4em, yshift=-1.2em] NW);
    \coordinate (SE) at ([xshift=1.2em, yshift=-1.2em] NW);
    \coordinate (NE) at ([xshift=1.2em, yshift=-0.4em] NW);
    \draw (SW) -- (SE) -- (NE);
  \end{tikzpicture}
  \qquad\qquad
\begin{tikzcd}[row sep=tiny]
	SX && GX \\
	& {SX'} && {GX'} \\
	SY && GY \\
	& {SY'} && {GY'}
	\arrow["Sf"', from=1-1, to=3-1]
	\arrow["Gf"{description, pos=0.2}, from=1-3, to=3-3]
	\arrow["{m_X}", from=1-1, to=1-3]
	\arrow["{m_Y}"{description, pos=0.8}, from=3-1, to=3-3]
	\arrow["Gx", from=1-3, to=2-4]
	\arrow["{Gf'}", from=2-4, to=4-4]
	\arrow["Gy"{description}, from=3-3, to=4-4]
	\arrow["{m_{Y'}}"', from=4-2, to=4-4]
	\arrow["Sy"', from=3-1, to=4-2]
	\arrow["Sx"{description}, from=1-1, to=2-2]
	\arrow["{Sf'}"{description, pos=0.3}, from=2-2, to=4-2, crossing over]
	\arrow["{m_{X'}}"{description, pos=0.8}, from=2-2, to=2-4, crossing over]
\end{tikzcd}
 \vspace*{-6mm} \]
\end{proof}

\begin{lemma}\label{lemma:restrict-cartesian}
  If $\A$ has pullbacks, then any
  factorization system on $\CartNt{\A}{\D}$ (all functors, but only
  cartesian natural transformations) restricts to a factorization
  system on $\Cart{\A}{\D}$ (cartesian functors and cartesian natural
  transformations).
\end{lemma}

\begin{proof}
  It suffices to show that $\Cart{\A}{\D}$ is closed under
  factorizations of cartesian natural transformations.
  This a consequence of the previous lemma:
  if $\tau : F \natto G$ is a morphism in $\Cart{\A}{\D}$ that
  factorizes as $(S, e, m)$, then $S$ is cartesian because $G$ and $m$
  are.
\end{proof}  

\begin{definition}\label{def:stable}
  If $\D$ has pullbacks, then a factorization system $(\E,\M)$ on $\D$
  is \emph{stable} when $\E$ is closed under pullbacks along arbitrary
  morphisms.
\end{definition}
(The analogous property for $\M$ is true in every factorization system.)

\begin{lemma}\label{lemma:compwise-lift-restrict}
  Assume that $\D$ has pullbacks, and let $(\E, \M)$ be a stable
  factorization system on $\D$.
  The componentwise lifting of $(\E,\M)$ to $[\A,\D]$ restricts to a
  factorization system on $\CartNt{\A}{\D}$.
\end{lemma}

\begin{proof}
  We need to check that, if a cartesian natural transformation
  $\tau : F \natto G$ factorizes as $(S, e, m)$ using $(\E,\M)$
  componentwise, then $e$ and $m$ are cartesian natural transformations.
  Given any $f : X \to Y$, we can consider the naturality square of
  $\tau$ for $f$, which is by assumption cartesian. It breaks into
  naturality squares of $e$ and $m$ for $f$. We can then pull back
  $(SY, m_Y)$ along $G f$ and be certain that the resulting morphism is
  in $\M$. The unique morphism from $FX$ to the
  pullback vertex $\bullet$ is a pullback of $(F Y, e_Y)$ along $Sf$ by the
  pullback lemma, and is therefore in $\E$ by stability.
  We therefore have two factorizations of $\tau_X$: one through $SX$ and
  one through $\bullet$.
  Factorizations are unique up to isomorphism, and hence the
  naturality squares of both $e$ and $m$ are pullbacks.
\[
\begin{tikzcd}[row sep=small]
  F X \arrow[dd, "Ff"']
    \arrow[r, "e_X", two heads] \arrow[rr, bend left=30, "\tau_X"]
    \arrow[dr, dashed, two heads]
  & S X \arrow[r, "m_X", tail] \arrow[dd, bend right=40, "Sf"']
    & G X \arrow[dd, "Gf"]  \\
  & \bullet \arrow[d] \arrow[ur, tail] & \\
  F Y \arrow[r, "e_Y", two heads] \arrow[rr, bend right=30, "\tau_Y"'] 
  & S Y \arrow[r, "m_Y", tail]
    & G Y  
\end{tikzcd}  
\]

We also need to check that the diagonal fill-ins of cartesian squares
built using $(\E,\M)$ componentwise are cartesian.
This holds because the pullback lemma provides a two-out-of-three
property for cartesian natural transformations: if $m \compose d$ and
$m$ are cartesian, then $d$ is also cartesian.
\end{proof}

\cref{lemma:compwise-lift-restrict} enables us to restrict the componentwise lifting of a
factorization system to $\CartNt{\A}{\D}$.
The following lemma enables us to restrict the factorization system
$(\E', \M')$ defined at the beginning of this section.

\begin{lemma}\label{lemma:cart-mono-restrict}
  Let $(\E, \M)$ be any factorization system on $[\A,\D]$. If all
  natural transformations in $\M$ are cartesian, then $(\E,\M)$
  restricts to a factorization system on $\CartNt{\A}{\D}$.
\end{lemma}

\begin{proof}
  If $\tau = m \compose e$ and $\tau$ and $m$ are cartesian, then $e$
  is cartesian by the pullback lemma.
\end{proof}  

Lemmata \ref{lemma:compwise-lift-restrict} and \ref{lemma:cart-mono-restrict} provide two constructions of a factorization system on
$\CartNt{\A}{\D}$: the componentwise lifting of a factorization system
on $\D$, and the factorization system $(\E', \M')$ of
\cref{lemma:partly-cartesian-factorizations}.
We now show that the two coincide.

\begin{proposition}\label{prop:coincide}
  Let $(\E, \M)$ be a stable factorization system on a category $\D$ with
  pullbacks, and let $\A$ be a category with a terminal object.
  The factorization system $(\E',\M')$ on $[\A,\D]$ from
  \cref{lemma:partly-cartesian-factorizations} and the componentwise
  lifting of $(\E,\M)$ to $[\A,\D]$ both restrict to the same
  factorization system on $\CartNt{\A}{\D}$.
\end{proposition}
\begin{proof}
  By stability of $(\E, \M)$, for each cartesian natural transformation
  $e$, having $e_1 \in \E$ is equivalent to having $e_X \in \E$ for all
  $X$.
  Hence when restricted to $\CartNt{\A}{\D}$, the $\E'$-morphisms are
  exactly the $\CartNt{\A}{\D}$-morphisms whose components are in $\E$,
  and similarly for $\M'$.
\end{proof}

When $\A$ has pullbacks it follows, using
\cref{lemma:restrict-cartesian}, that the two factorization systems also
restrict to the same factorization system on $\Cart{\A}{\D}$.
When $\A = \D$, the latter forms a monoidal category with functor
composition as tensor, and $\E'$ is closed under $({-}) \fcomp S$ for
every cartesian endofunctor $S$.
Hence we can construct canonical gradings of cartesian
monads (monoids in $\Cart{\D}{\D}$) by \cref{thm:canonical-grading}.
The grades of the canonical grading of a cartesian monad $\T$ are
equivalently $\M$-subobjects of $T1$, by
\cref{lemma:partly-cartesian-factorizations}.

\begin{example}\label{ex:list-cartesian}
  Let us return to the factorization system $(\E, \M) = (\Surj, \Inj)$ on $\A = \D = \Set$, which is stable. Let $(\E',\M')$ be the factorization system on $\Cart{\Set}{\Set}$ just discussed. Then $\E'$ is closed both under $(-) \fcomp S$ and under $S \fcomp (-)$ for any cartesian set functor $S$. Hence $\slice{\M'}{T}$ acquires a monoidal structure for any cartesian set monad $\T$.
  
  Let $\T$ be the list monad on $\Set$, so $TX$ is the set of lists over
  $X$, the unit is $\eta_X\,x = [x]$, and the multiplication is
  $\mu_X [\xs_1, \dots, \xs_n] = \xs_1 \append \cdots \append \xs_n$,
  where $(\append)$ is concatenation of lists.
  This monad is cartesian.
  There is an isomorphism $T1 \iso \Nat$, so shapes are equivalently
  natural numbers (corresponding to the length of the list).
  Then $\M'$-subobjects of $T$ are equivalently subsets
  $\Sigma \subseteq \Nat$.
  By the above, these form the canonical $\M'$-grading of $\T$.
  The monoidal structure on these subsets is given by
  \[\textstyle
    \cI = \{1\}
    \qquad
    \Sigma \ctensor \Sigma' = \{ \sum_{i = 1}^n m_i \mid n \in \Sigma, m_1,
    \dots, m_n \in \Sigma'\}
  \]
  The graded monad $\slicedom{\M'}{\T}$ is given on objects by
  $
    \slicedom{\M'}{\T} \Sigma\,X = \{ \xs \mid |\xs| \in \Sigma \}
  $,
  where $|\xs|$ is the length of $\xs$.
\end{example}

\section{Algebraic operations}\label{sec:algebraic-operations}

In models of computational effects, we usually do not just want a (strong)
monad $\T$; we also want to equip $\T$ with a collection of \emph{algebraic
operations} in the sense of Plotkin and Power~\cite{plotkin2003algebraic}.
The latter provide interpretations of the constructs that cause
the effects.
When modelling computations using a graded monad, we similarly want
algebraic operations for the graded monad; such a notion of algebraic
operation was introduced in \cite{KMUW:flepgm}.
In this section, we therefore investigate the problem of constructing
algebraic operations for the graded monad $\slicedom{\M}{\T}$, given algebraic
operations for the monad $\T$.

Throughout this section, we assume a monoidal category
$\monoidal{C} = (\C, \tensorI, \tensor)$ that has finite products, for
example, endofunctors on a category with finite products.
When we write $T^n$ below, we mean the product of $n$-many copies of
$T$.
We work only with normal (i.e., non-skew) monoidal categories in this
section.
The notion of algebraic operation for a graded monoid (e.g., a
graded monad) that we use below works for monoidal categories,
but the appropriate notion for skew monoidal categories would be more
complicated.
(It would use a list of grades $e_i$ instead of a single grade $e$ in the definition below.)
Hence when we consider the canonical gradings below, we work under the
assumption that they form a monoidal category (for example, when
$\E$ is closed under $\tensor$ in both arguments).

The following definition generalizes the notion of algebraic operation
for a monad to monoids.
\begin{definition}
  Let $\T = (T, \eta, \mu)$ be a monoid in $\monoidal{C}$.
  An $n$-ary \emph{algebraic operation} for $\T$, where $n$ is a natural
  number, is a morphism
  $
    \oper : T^n \to T
  $
  such that
  \[
    \begin{tikzcd}[column sep=large]
      T^n \tensor T
      \arrow[d, "\oper \tensor T"']
      \arrow[r, "\langle \pi_i \tensor T \rangle_i"] &
      (T \tensor T)^n
      \arrow[r, "\mu^n"] &
      T^n
      \arrow[d, "\oper"] \\
      T \tensor T \arrow[rr, "\mu"] &&
      T
    \end{tikzcd}
  \]
\end{definition}

\begin{definition}
  Let
  $\lax G = (G, \eta, \mu) : \monoidal{G} \to \monoidal{C}$
  be a $\monoidal{G}$-graded monoid in $\monoidal{C}$.
  A $(d_1, \dots, d_n; d')$-ary algebraic operation for $\lax{G}$, where
  $d_1, \dots, d_n, d' \in \monoidal{G}$, is a natural transformation
  $
    \goper_e : \prod_i G (d_i \gtensor e) \natto G (d' \gtensor e)
  $
  such that, for all $e, e' \in \G$,
  \[
    \begin{tikzcd}
      (\prod_i G(d_i \gtensor e)) \tensor Ge'
      \arrow[r, "\langle \pi_i \tensor Ge' \rangle_i"]
      \arrow[d, "\goper_e \tensor Ge'"'] &
      \prod_i (G(d_i \gtensor e) \tensor Ge')
      \arrow[r, "\prod_i \mu_{d_i \gtensor e, e'}"] &
      \prod_i G((d_i \gtensor e) \gtensor e')
      \arrow[r, "G \alpha"] &
      \prod_i G(d_i \gtensor (e \gtensor e'))
      \arrow[d, "\goper_{e \gtensor e'}"] \\
      G(d' \gtensor e) \tensor Ge'
      \arrow[rr, "\mu_{d' \gtensor e, e'}"] &&
      G((d' \gtensor e) \gtensor e')
      \arrow[r, "G \alpha"] &
      G(d' \gtensor (e \gtensor e'))
    \end{tikzcd}
  \]
\end{definition}

\begin{example}
  Let $\monoidal C$ be the cartesian monoidal category $\Set$.
  Then an $n$-ary algebraic operation for a monoid $\T$ is a function
  $\oper : T^n \to T$ such that the multiplication of the monoid
  distributes over $\oper$ from the right.
  For example, if $\T$ is natural numbers with ordinary multiplication,
  then $\phi (x_1, \dots, x_n) = x_1 + \cdots + x_n$ is an $n$-ary
  algebraic operation.
\end{example}

\begin{definition}
  Let $(\monoidal{G}, \lax{G}, g)$ be an $\M$-grading of a monoid $\T$,
  where $\M$ is a class of morphisms in $\C$.
  We say that a $(d_1, \dots, d_n; d')$-ary algebraic operation
  $\goper$ for $\lax{G}$ is a \emph{grading} of an $n$-ary algebraic
  operation $\oper$ for $\T$ when the following diagram commutes for all
  $e \in \monoidal{G}$.
  \[
    \begin{tikzcd}
      \prod_i G(d_i \gtensor e)
      \arrow[r, "\goper_e"]
      \arrow[d, "\prod_i g_{d_i \gtensor e}"'] &
      G(d' \gtensor e)
      \arrow[d, "g_{d' \gtensor e}", tail] \\
      T^n
      \arrow[r, "\oper"] &
      T
    \end{tikzcd}
  \]
\end{definition}

Suppose that $\T$ is a monoid in $\monoidal{C}$, and that
$(\E, \M)$ is a factorization system on $\C$ such that
$\E$ is closed under $({-}) \tensor S$ for all $S \monoto T$.
Then $\T$ has a canonical grading
$\slicedom{\M}{\T} : \slice{\M}{\T} \to \monoidal{C}$ by
\cref{thm:canonical-grading}.
Suppose in addition that the skew monoidal category $\slice{\M}{\T}$ is
actually monoidal (which is the case when $\E$ is closed also under
$S \tensor {({-})}$ for all $S \monoto T$).
We keep these assumptions without repeating them for the rest of this
section.

Our goal in the rest of this section is to show that we can assign
canonical grades to algebraic operations for $\T$.
To be more precise, let $\oper : T^n \to T$ is an $n$-ary algebraic
operation for $\T$, and let $R_1, \dots, R_n$ be a list of grades
($\M$-subobjects of $T$).
We show how to construct a grade $R'$ and an algebraic operation
\[\textstyle
  \goper : \prod_i \slicedom{\M}{\T} (R_i \ctensor {-}) \to \slicedom{\M}{\T} (R' \ctensor {-})
\]
of arity $(R_1, \ldots, R_n; R')$ for $\slicedom{\M}{\T}$,
such that $\goper$ grades $\oper$.
The grade $R'$ is in a sense canonical (see
\cref{thm:grading-operations} below), and in fact every component of
$\goper$ is in $\E$.

To do this, we make the following two further assumptions about $\E$ for
the rest of the section.
Firstly, we assume that $\E$ contains the canonical morphisms
$
  \langle \pi_i \tensor Y \rangle_i
    : (\prod_i X_i) \tensor Y \to \prod_i (X_i \tensor Y)
$.
This is the case in particular when $\tensor$ preserves finite products on
the left (because $\E$ contains all isomorphisms); when $\tensor$ is
composition of endofunctors this is automatically true.
Secondly, we assume that $\E$ is closed under finite products, i.e.\
that $\prod_i e_i : \prod_i X_i \to \prod_i Y_i$ is in $\E$ whenever all of the
morphisms $e_i : X_i \epito Y_i$ are in $\E$.
This is the case for all of the factorization systems we consider above.

The key lemma that enables us to construct $\goper$ is the following,
which characterizes algebraic operations for the canonical grading
$\slicedom{\M}{\T}$ of $\T$.

\begin{lemma}\label{lemma:canonical-operations}
  Let $\oper : T^n \to T$ be an $n$-ary algebraic operation for $\T$,
  and let $R_1, \dots, R_n, R'$ be $\M$-subobjects of $T$.
  There is a bijection between (1) morphisms $p : \prod_i R_i \to R'$ such
  that
  \[
    \begin{tikzcd}[row sep=small]
      \prod_i R_i \arrow[r, "p"] \arrow[d, tail] &
      R' \arrow[d, tail] \\
      T^n \arrow[r, "\oper"] &
      T
    \end{tikzcd}
  \]
  and (2) $(R_1, \dots, R_n; R')$-ary algebraic operations $\goper$ for
  $\slicedom{\M}{\T}$ that grade $\oper$.
\end{lemma}
\begin{proof}
  Given a morphism $p$ as in (1), the following square commutes because
  $\oper$ is algebraic, and the square hence has a unique diagonal
  $\goper_S$.
  Further applications of orthogonality show that $\goper$ is an algebraic
  operation.
  It is a grading of $\oper$ by definition.
  \[
    \begin{tikzcd}[row sep=small]
      (\prod_i R_i) \tensor S
      \arrow[r, "\langle \pi_i \tensor S \rangle_i", two heads]
      \arrow[d, "p \tensor S"'] &
      \prod_i (R_i \tensor S)
      \arrow[r, "\prod_i q_{R_i, S}", two heads] &
      \prod_i (R_i \ctensor S)
      \arrow[d, tail]
      \arrow[lldd, "\goper_S"{description}, dashed] \\
      R' \tensor S
      \arrow[d, "q_{R', S}"', two heads] &&
      T^n
      \arrow[d, "\oper"] \\
      R' \ctensor S
      \arrow[rr, tail] &&
      T
    \end{tikzcd}
  \]
  In the other direction, given $\goper$, we have a morphism
  $p$ as follows; this $p$ makes the diagram required for (1) commute
  because $\goper$ is a grading of $\oper$.
  \[\textstyle
    p : \prod_i R_i \xto{\prod_i r_{R_i}} \prod_i (R_i \ctensor \cI)
      \xto{\goper_{\cI}} R' \ctensor \cI \xto{r^{-1}_{R'}} R'
  \]
  From algebraicity of $\goper$ it follows that this $p$ makes the
  upper triangle of the above square commute and hence, by uniqueness of
  the diagonal, that $\goper$ is the only grading of $\oper$ that
  induces this $p$.
  The construction of $p$ from $\goper$ is therefore injective.
  The following diagram chase shows that constructing a new $p$ from the $\goper$ constructed from a given $p$ yields the same $p$, hence the constructions
  form a bijection. 
  \[
\begin{tikzcd}[row sep=0.8em, column sep=6em]
	{\prod_i R_i} \\
	& {(\prod_i R_i) \tensor \cI} & {\prod_i(R_i \tensor \cI)} & {\prod_i(R_i \ctensor \cI)} \\
	& {R' \tensor \cI} && {R' \ctensor \cI} \\
	{R'}
	\arrow["p"', from=1-1, to=4-1]
	\arrow["{\prod_i r_{R_i}}", curve={height=-12pt}, from=1-1, to=2-4]
	\arrow["{\psi_{\cI}}", from=2-4, to=3-4]
	\arrow["{r_{R'}}"', curve={height=12pt}, from=4-1, to=3-4]
	\arrow["{p \tensor \cI}"', from=2-2, to=3-2]
	\arrow["{\langle \pi_i \tensor \cI\rangle_i}"', from=2-2, to=2-3]
	\arrow["{\prod_i q_{R_i, \cI}}"', from=2-3, to=2-4]
	\arrow["{q_{R', \cI}}", from=3-2, to=3-4]
	\arrow["{(R' \tensor q) \compose \rho_{R'}}" near end, from=4-1, to=3-2]
	\arrow["{(\prod_i R_i \tensor q) \compose \rho_{\prod_i R_i}}"' very near end, from=1-1, to=2-2]
\end{tikzcd}
 \vspace*{-6mm} \]
\end{proof}

Now given an $n$-ary algebraic operation $\oper$ for $\T$ and a fixed
tuple $R_1, \dots, R_n$ of $\M$-subobjects of $T$, we construct the
canonical $R'$ by factorizing
\begin{tikzcd}[cramped, column sep=1.6em]
  \prod_i R_i \arrow[r, tail] &
  T^n \arrow[r, "\oper"] &
  T
\end{tikzcd}
as
\begin{tikzcd}[cramped, column sep=1.6em]
  \prod_i R_i \arrow[r, "p", two heads] &
  R' \arrow[r, tail] &
  T
\end{tikzcd}.
The preceding lemma then provides us with an
$(R_1, \dots, R_n; R')$-algebraic operation $\goper$ for $\slicedom{\M}{\T}$.

\begin{theorem}\label{thm:grading-operations}
  Let $\phi : T^n \to T$ be an $n$-ary algebraic operation for $\T$.
  \begin{enumerate}
    \item The construction above defines an $(R_1, \dots, R_n; R')$-ary
      algebraic operation $\goper$ for $\slicedom{\M}{\T}$, and $\goper$ grades
      $\oper$.
      Every component $\goper_S$ is in $\E$.
    \item For any $\M$-subobject $R'' \monoto T$ and
      $(R_1,\dots,R_n; R'')$-ary algebraic operation $\goper'$ for
      $\slicedom{\M}{\T}$, such that $\goper'$ grades $\oper$,
      there is a unique $f : R' \to R''$ in $\slice{\M}{T}$
      such that $(f \ctensor S) \compose \goper_S = \goper'_S$ for
      all $S$.
  \end{enumerate}
\end{theorem}
\begin{proof}
  The first sentence of (1) is immediate from
  \cref{lemma:canonical-operations}.
  Each $\goper_S$ is in $\E$ because we have
  $\goper_S \compose e = e'$ for some $e, e' \in \E$ (this is the upper
  triangle in the definition of $\goper_S$, using the fact that
  $p$ is in $\E$).
  This implies $\goper_S \in \E$ because $\E$-morphisms satisfy a
  two-out-of-three property.
  For (2), given $\goper'$, we obtain from
  \cref{lemma:canonical-operations} a morphism $p' : \prod_i R_i \to
  R''$ making the diagram on the left below commute.
  \[
    \begin{tikzcd}
      \prod_i R_i \arrow[r, "p'"] \arrow[d, tail] &
      R'' \arrow[d, tail] \\
      T^n \arrow[r, "\oper"] &
      T
    \end{tikzcd}
    \qquad
    \begin{tikzcd}[column sep=large]
      \prod_i R_i
      \arrow[r, "p", two heads]
      \arrow[d, "p'"'] &
      R'
      \arrow[d, tail]
      \arrow[ld, "f"{description}, dashed] \\
      R''
      \arrow[r, tail] &
      T
    \end{tikzcd}
  \]
  For a morphism $f : R' \to R''$ in $\slice{\M}{T}$, the condition that
  $(f \ctensor S) \compose \goper_S = \goper'_S$ for all $S$ implies
  (using $S = \cI$) that $f \compose p = p'$.
  The converse also holds, using orthogonality.
  Hence the conditions on the morphism $f$ are equivalent to
  commutativity of the square on the right above.
  The outside of the square commutes and $p$ is in $\E$, so there exists a
  unique $f$.
\end{proof}

\begin{example}
  Consider the writer monad given by $T = M \times ({-})$ from
  \cref{ex:writer-monoidal}.
  Every $z \in M$ induces a unary algebraic operation
  $\oper_z : T \to T$, defined by
  $
    \oper_{z,X} (z', x) = (z \cdot z', x)
  $.
  When $\M$ is the class of componentwise injective natural
  transformations, the canonical $\M$-grading of $\T$ has
  subsets $\Sigma \subseteq M$ as grades, and
  $\slicedom{\M}{\T} \Sigma = \Sigma \times ({-})$.
  Every input grade $P \subseteq M$ induces a canonical output grade
  $P'_z \subseteq M$ and algebraic operation
  $
    \goper_{z, \Sigma}
      : \slicedom{\M}{\T} (P \ctensor \Sigma) \natto \slicedom{\M}{\T} (P'_z \ctensor \Sigma)
  $,
  and these turn out to be:
  \[
    P'_z = \{z \cdot z' \mid z' \in P\}
    \qquad
    \goper_{z, \Sigma, X} (z', x) = (z \cdot z', x)
  \]
\end{example}

\begin{example}
  Let $\T$ be the list monad on $\Set$.
  This has a binary algebraic operation
  $(\append) : T \times T \natto T$ that concatenates a pair of lists.
  As we explain in \cref{ex:list-cartesian}, subsets $\Sigma \subseteq \Nat$
  provide a canonical grading of $\T$.
  If $P_1, P_2$ are subsets of $\Nat$, then the grade we construct for
  the algebraic operation $(\append)$ as above is
  $P' = \{ n_1 + n_2 \mid n_1 \in P_1, n_2 \in P_2 \}$, and the
  algebraic operation for $\slicedom{\M}{\T}$ is
  the natural transformation
  $
    \slicedom{\M}{\T}(P_1 \ctensor {-}) \times \slicedom{\M}{\T}(P_2 \ctensor {-})
      \natto \slicedom{\M}{\T}(P' \ctensor {-})
  $
  that maps $(\xs_1, \xs_2)$ to $\xs_1 \append \xs_2$.
\end{example}

\section{Conclusion and future work}

We have demonstrated that factorization systems provide a unifying
framework for the grading of monads by subfunctors, in fact, monoids
with subobjects. Skew monoidal categories turn out to be a more robust
setting for this than monoidal categories, which means, among other
things, that this framework will be directly applicable also to
relative monads.

The abstract framework is pleasingly elegant, but for applications we
would like obtain a stronger intuition for its reach.  We intend to
explore this first by working out the canonical gradings with (strong)
subfunctors of further standard example (strong) monads from
programming semantics, for the factorization systems considered in
this paper and possibly others. Indeed, the examples may point to
further factorization systems of interest. The outcomes of this
exploration will hopefully lead to some new heuristics for the
construction of graded monads for applications such as type-and-effect
systems.

Programming semantics applications also suggest trying grading with
subfunctors on (strong) lax monoidal functors (``applicative
functors'') and (strong) monads in $\mathbf{Prof}$
(``arrows''). Comonads can be graded with quotient functors.


\paragraph{Acknowledgements} D.M.\ and T.U.\ were supported by the
Icelandic Research Fund grants 196323-053 and 228684-051. This work
started during T.U.'s visit to F.B.\ at LIPN as an invited professor
in Sept.\ 2019, but was mostly stalled during the Covid pandemic.


\bibliographystyle{eptcs}
\bibliography{act22}

\clearpage

\appendix

\section{Proof of Theorem~\ref{thm:slice-mon}}
\label{sec:proofs}

Given a monoidal category $(\C, \I, \ot, \lambda, \rho, \alpha)$ with
a monoid object $(T, \eta, \mu)$ and an orthogonal factorization
system $(\E, \M)$. We assume that $\E$ is closed under $(-) \ot X$ for
all $(X, x) \in \M / T$.

Our aim is to show $\M / T$ carries a left-skew monoidal category
structure $((\J, j), \bt, \ell, r, a)$.

The unit $(\J, j)$ and tensor $(X \bt Y, x \bt y)$ of two objects
$(X, x)$, $(Y, y)$ are defined as the factorizations shown in the
diagrams below.

\[
\begin{tikzcd}[row sep=small]
  \I \arrow[rr, "\eta"] \arrow[dr, "q"', two heads]
  & & T \\
  & \J \arrow[ur, "j"', tail]
\end{tikzcd}
\qquad
\begin{tikzcd}[row sep=small]
  X \ot Y \arrow[r, "x \ot y"] \arrow[dr, "q_{x,y}"', two heads]
  & T \ot T \arrow[r, "\mu"]
    & T \\
  & X \bt Y \arrow[ur, "x \bt y"', tail]
\end{tikzcd}
\]


The functorial action of $\bt$ on two morphisms
$f : (X, x) \to (X', x')$ and $g : (Y, y) \to (Y', y')$ is a morphism
$f \bt g : (X \bt Y, x \bt y) \to (X' \bt Y', x' \bt y')$ defined as
the diagonal fill-in of the commuting square below.
\[
\begin{tikzcd}[column sep=large]
  X \ot Y \arrow[dd, "f \ot g"'] \arrow[dr, "x \ot y"] \arrow[r, "q_{x,y}", two heads]
  &  X \bt Y \arrow[dd, bend left=50, dashed, "f \bt g"' near end, tail] \arrow[dr, "x \bt y", tail] \\
  & T \ot T \arrow[r, "\mu"]
    & T \\
  X' \ot Y' \arrow[ur, "x' \ot y'"'] \arrow[r, "q_{x',y'}"', two heads]
  & X' \bt Y' \arrow[ur, "x' \bt y'"', tail]
\end{tikzcd}
\]

The left unitor $\ell$ and associator $a$ are also defined as the diagonal
fill-ins for suitable commuting squares. The right unitor is just a
composition of morphisms.

Definition of $\ell$:
\[
\begin{tikzcd}
  \I \ot X \arrow[dd, "\lambda_X"'] \arrow[dr, "\I \ot x"'] \arrow[r, "q \ot X", two heads]
  & \J \ot X \arrow[dr, "j \ot x"] \arrow[r, "q_{j,x}", two heads]
    & \J \bt X \arrow[dr, "j \bt x", tail] \arrow[lldd, bend right=22, dashed, "\ell_x"', tail] \\
  & \I \ot T \arrow[d, "\lambda_T"'] \arrow[r, "\eta \ot T"]
    & T \ot T \arrow[r, "\mu"]
      & T \\
  X \arrow[r, "x"', tail]
  & T \arrow[urr, equals]
\end{tikzcd}
\]

Definition of $r$:
\[
\begin{tikzcd}
  X \arrow[ddrr, bend right=22, dashed, "r_x"', tail] \arrow[dd, "\rho_X"'] \arrow[r, "x", tail] 
  & T \arrow[d, "\rho_T"'] \arrow[drr, equals] \\
  & T \ot \I \arrow[r, "T \ot \eta"'] 
    & T \ot T \arrow[r, "\mu"']
  & T \\
  X \ot \I \arrow[ur, "x \ot \I"] \arrow[r, "X \ot q"']  
  & X \ot \J \arrow[ur, "x \ot j"'] \arrow[r, "q_{x,j}"', two heads]
    & X \bt \J \arrow[ur, "x \bt j"', tail]
\end{tikzcd} 
\]

Definition of $a$:
\[
\begin{tikzcd}[row sep=small]
  (X \ot Y) \ot Z \arrow[dddddd, "\alpha_{X,Y,Z}"'] \arrow[ddr, "(x \ot y) \ot z"']
                                       \arrow[r, "q_{x,y} \ot Z", two heads]
   & (X \bt Y) \ot Z \arrow[ddr, "(x \bt y) \ot z"]
                                       \arrow[r, "q_{x \bt y,z}", two heads] 
     & (X \bt Y) \bt Z \arrow[dddddd, bend right=37, dashed, "a_{x,y,z}"]
                                       \arrow[dddr, "(x \bt y) \bt z", tail]
                                                                    \\ \\
  & (T \ot T) \ot T \arrow[dd, "\alpha_{T,T,T}"'] \arrow[r, "\mu \ot T"']
    & T \ot T \arrow[dr, "\mu"'] \\
  & & & T \\
  & T \ot (T \ot T) \arrow[r, "T \ot \mu"]
    & T \ot T \arrow[ur, "\mu"] \\ \\
  X \ot (Y \ot Z) \arrow[uur, "x \ot (y \ot z)"] \arrow[r, "X \ot q_{y,z}"']
  & X \ot (Y \bt Z) \arrow[uur, "x \ot (y \bt z)"'] \arrow[r, "q_{x, y\bt z}"', two heads]
    & X \bt (Y \bt Z) \arrow[uuur, "x \bt (y \bt z)"', tail]
\end{tikzcd}
\]

The proofs of functoriality of $\bt$ and naturality of $\ell$, $r$ and
$a$ are easy and omitted. 






The equations (m1)--(m5) for $\ell$, $r$, $a$ are each proved from the
respective equations of $\lam$, $\rho$, $\al$ using the properties of
$\bt$, $\ell$, $r$, $a$ arising from their construction (the two
triangles that the fill-in breaks the square into). For each equation
$\mathit{lhs} = \mathit{rhs}$, the two sides $\mathit{lhs}$ and $\mathit{rhs}$ are both shown to be the
diagonal fill-in of a square of the form
$s \compose e = s' \compose f$ where $e$ is a suitable $\E$-morphism
and $s$ and $s'$ are the common domain resp.\ codomain of $\mathit{lhs}$ and
$\mathit{rhs}$ as morphisms in $\M/T$. Below are the diagram chases for the
triangles $\mathit{lhs} \compose e = f = \mathit{rhs} \compose e$; the triangles
$s' \compose \mathit{lhs} = s = s' \compose \mathit{rhs}$ are straightforward.

Proof of (m1):
\[
\begin{tikzcd}
  \I \arrow[dd, equals] \arrow[dr, "\rho_\I"'] \arrow[r, "q", two heads]      
  & \J \arrow[r, "\rho_\J"'] \arrow[drrr, bend left=30, "r_j", tail]
    & \J \ot \I \arrow[dr, "\J \ot q"]
      & \\
  & \I \ot \I \arrow[dl, "\lambda_\I"'] \arrow[ur, "q \ot \I"', two heads] \arrow[dr, "\I \ot q"]
    & & \J \ot \J \arrow[r, "q_{j,j}"', two heads]
        & \J \bt \J \arrow[dlll, bend left=30, "\ell_j", tail] \\
  \I \arrow[r, "q", two heads]        
  & \J
    & \I \ot \J \arrow[l, "\lambda_\J"'] \arrow[ur, "q \ot \J"', two heads]
    &
\end{tikzcd}
\]  

Proof of (m2):
\[
\begin{tikzcd}[column sep=large]
  & (\J \ot Y) \ot Z  \arrow[d, "\alpha_{\J,Y,Z}"] \arrow[r, "q_{j,y} \ot Z", two heads]
    & (\J \bt Y) \ot Z \arrow[r, "q_{j \bt y, z}", two heads]
      & (\J \bt Y) \bt Z \arrow[d, "a_{j,y,z}", tail] \\
  & \J \ot (Y \ot Z)  \arrow[r, "\J \ot q_{y,z}"]
    & \J \ot (Y \bt Z)  \arrow[r, "q_{j, y \bt z}", two heads]
      & \J \bt (Y \bt Z) \arrow[dd, "\ell_{y \bt z}", tail]\\
  (\I \ot Y) \ot Z \arrow[r, "\alpha_{\I,Y,Z}"] \arrow[uur, "(q \ot Y) \ot Z", two heads] \arrow[dr, "\lambda_Y \ot Z"] \arrow[ddr, "(q \ot Y) \ot Z"', two heads]     
  & \I \ot (Y \ot Z) \arrow[u, "q \ot (Y \ot Z)"', two heads] \arrow[d, "\lambda_{Y \ot Z}"] \arrow[r, "\I \ot q_{y,z}"]
    & \I \ot (Y \bt Z) \arrow[u, "q \ot (Y \bt Z)"', two heads] \arrow[dr, "\lambda_{Y \bt Z}"]
      & \\
  & Y \ot Z \arrow[rr, "q_{y,z}", two heads]
    & & Y \bt Z \\  
  & (\J \ot Y) \ot Z \arrow[r, "q_{j,y} \ot Z"', two heads]
    & (\J \bt Y) \ot Z \arrow[ul, "\ell_y \ot Z"'] \arrow[r, bend left=10, "q_{y \circ \ell_y,z}", two heads] \arrow[r, bend right=10, "q_{j \bt y, z}"', two heads]
      & (\J \bt Y) \bt Z \arrow[u, "\ell_y \bt Z"', tail] \\
\end{tikzcd}
\]

Proof of (m3):
\[
\begin{tikzcd}[row sep=small,column sep=large]
  & X \bt Y \arrow[r, "\rho_{X \bt Y}"] \arrow[rrr, bend left=15, "r_{x \bt y}", tail]
    & (X \bt Y) \ot \I  \arrow[r, "(X \bt Y) \ot q"]
      & (X \bt Y) \ot \J \arrow[r, "q_{x \bt y, j}", two heads]
        & (X \bt Y) \bt \J \arrow[dddddd, "a_{x,y,j}", tail]\\ \\
  & (X \ot Y) \ot \I \arrow[dd, "\alpha_{X,Y,\I}"] \arrow[uur, "q_{x,y} \ot \I"', two heads]
                                            \arrow[r, "(X \ot Y) \ot q"']     
    & (X \ot Y) \ot \J \arrow[dd, "\alpha_{X,Y,\J}"] \arrow[uur, "q_{x,y} \ot \J"']
      & & \\
  X \ot Y \arrow[dddr, bend right=10, "q_{x,y}"', two heads] \arrow[dddr, bend left=10, "q_{x,y \bt j \circ q_{y,j} \circ Y \ot q \circ \rho_Y}" very near end, two heads] \arrow[dr, "X \ot \rho_Y"] \arrow[ur, "\rho_{X \ot Y}"'] \arrow[uuur, "q_{x,y}", two heads] \\
  & X \ot (Y \ot \I) \arrow[r, "X \ot (Y \ot q)"] 
    & X \ot (Y \ot \J) \arrow[r, "X \ot q_{y,j}"]
      & X \ot (Y \bt \J) \arrow[ddr, "q_{x, y \bt j}", two heads]
        & \\ \\
  & X \bt Y \arrow[r, "X \bt \rho_Y"] \arrow[rrr, bend right=15, "X \bt r_y"'] 
    & X \bt (Y \ot \I) \arrow[r, "X \bt (Y \ot q)"]
      & X \bt (Y \ot \J) \arrow[r, "X \bt q_{y,j}"]
        & X \bt (Y \bt \J)
\end{tikzcd}
\]  

Proof of (m4):
\[
\begin{tikzcd}
  & X \bt Z \arrow[r, "\rho_X \bt Z"'] \arrow[rrr, bend left=15, "r_x \bt Z"]
    & (X \ot \I) \bt Z \arrow[r, "(X \ot q) \bt Z"']
      & (X \ot \J) \bt Z \arrow[r, "q_{x,j} \bt Z"', two heads]
        & (X \bt \J) \bt Z \arrow[ddd, "a_{x,j,z}", tail] \\
  X \ot Z \arrow[d, equals] \arrow[r, "\rho_X \ot Z"'] \arrow[ur, bend left=10, "q_{x,z}", two heads]
          \arrow[ur, bend right=10, "q_{x \bt j \circ q_{x,j} \circ X \ot q \circ \rho_X,z}"' near end, two heads]      
  & (X \ot \I) \ot Z \arrow[d, "\alpha_{X,\I,Z}"] \arrow[r, "(X \ot q) \ot Z"]
    & (X \ot \J) \ot Z \arrow[d, "\alpha_{X,\J,Z}"'] \arrow[r, "q_{x,j} \ot Z", two heads]
      & (X \bt \J) \ot Z \arrow[ur, "q_{x \bt j,z}"', two heads]     
        & \\
  X \ot Z \arrow[dr, "q_{x,z}"', two heads]      
  & X \ot (\I \ot Z) \arrow[l, "X \ot \lambda_Z"'] \arrow[dr, bend right=10, "q_{x, z \circ \lambda_Z}"' near start, two heads] \arrow[dr, bend left=10, "q_{x,j \bt z \circ q_{j,z} \circ q \ot Z}" near end, two heads] \arrow[r, "X \ot (q \ot Z)"]
    & X \ot (\J \ot Z) \arrow[r, "X \ot q_{j,z}"]
      & X \ot (\J \bt Z) \arrow[dr, "q_{x, j \bt Z}", two heads]
        & \\
  & X \bt Z
    & X \bt (\I \ot Z) \arrow[l, "X \bt \lambda_Z"'] \arrow[r, "X \bt (q \ot Z)"]
      & X \bt (\J \ot Z) \arrow[r, "X \bt q_{j,z}"]
        & X \bt (\J \bt Z) \arrow[lll, bend left=15, "X \bt \ell_z"]
\end{tikzcd}
\]  

Proof of (m5):
\[
\scriptsize
\begin{tikzcd}[column sep=large]
  & ((X \bt Y) \ot Z) \ot W \arrow[dr, "\alpha_{X \bt Y,Z,W}"] \arrow[r, "q_{x \bt y, z} \ot W", two heads]
    & ((X \bt Y) \bt Z) \ot W \arrow[r, "q_{(x \bt y) \bt z, w}", two heads]
      & ((X \bt Y) \bt Z) \bt W \arrow[dr, "a_{x \bt y, z, w}", tail]
         \\
  & (X \ot Y) \ot (Z \ot W) \arrow[dd, "\alpha_{X,Y,Z \ot W}"] \arrow[dr, "(X \ot Y) \ot q_{z,w}"] \arrow[r, "q_{x,y} \ot (Z \ot W)", two heads]
    & (X \bt Y) \ot (Z \ot W)  \arrow[r, "(X \bt Y) \ot q_{z,w}"]    
      & (X \bt Y) \ot (Z \bt W) \arrow[r, "q_{x \bt y, z \bt w}", two heads]
        & (X \bt Y) \bt (Z \bt W) \arrow[dd, "a_{x,y,z \bt w}", tail]
       \\
  & & (X \ot Y) \ot (Z \bt W) \arrow[d, "\alpha_{X,Y,Z \bt W}"]
      \arrow[ur, "q_{x,y} \ot (Z \bt W)"', two heads]
      & \\    
  (X \ot Y) \ot Z) \ot W \arrow[uur, "\alpha_{X \ot Y,Z,W}"' near start] \arrow[uuur, "(q_{x,y} \ot Z) \ot W", two heads]
      \arrow[ddr, "\alpha_{X,Y,Z} \ot W" near start]
      \arrow[dddr, "(q_{x,y} \ot Z) \ot W"', two heads]     
  & X \ot (Y \ot (Z \ot W)) \arrow[r, "X \ot (Y \ot q_{z,w})"] \arrow[r, "X \ot (Y \ot q_{z,w})"]
    & X \ot (Y \ot (Z \bt W)) \arrow[r, "X \ot q_{y,z \bt w}"] \arrow[r, "X \ot q_{y, z \bt w}"]
      & X \ot (Y \bt (Z \bt W)) \arrow[r, "q_{x, y \bt (z \bt w)}", two heads]
        & X \bt (Y \bt (Z \bt W)) \\
  & X \ot ((Y \ot Z) \ot W) \arrow[u, "X \ot \alpha_{Y,Z,W}"'] \arrow[r, "X \ot (q_{y,z} \ot W)"]
     & X \ot ((Y \bt Z) \ot W) \arrow[r, "X \ot q_{y \bt z, w}"]
       & X \ot ((Y \bt Z) \bt W) \arrow[u, "X \ot a_{y,z,w}"] \arrow[r, "q_{x, (y \bt z) \bt w}", two heads]
         & X \bt ((Y \bt Z) \bt W) \arrow[u, "X \bt a_{y,z,w}"', tail] \\
  & (X \ot (Y \ot Z)) \ot W \arrow[u, "\alpha_{X,Y \ot Z,W}"'] \arrow[r, "(X \ot q_{y,z}) \ot W"] 
    & (X \ot (Y \bt Z)) \ot W \arrow[u, "\alpha_{X,Y \bt Z,W}"'] \arrow[r, "q_{x, y \bt z} \ot W", two heads]
      & (X \bt (Y \bt Z)) \ot W \arrow[r, "q_{x \bt (y \bt z), w}", two heads]
        & (X \bt (Y \bt Z)) \bt W \arrow[u, "a_{x,y\bt z,w}"', tail] \\
  & ((X \bt Y) \ot Z) \ot W \arrow[r, "q_{x \bt y, z} \ot W"', two heads]
    & ((X \bt Y) \bt Z) \ot W \arrow[ur, "a_{x,y,z} \ot W", tail] \arrow[r, "q_y{(x \bt y) \bt z, w}"', two heads]
      & ((X \bt Y) \bt Z) \bt W \arrow[ur, "a_{x,y,z} \bt W"', tail]
        &     
\end{tikzcd} 
\]  

\end{document}